\documentclass[twocolumn,english,aps,pra,floatfix,tightenlines,superscriptaddress,nofootinbib]{revtex4-1}
\usepackage[T1]{fontenc}
\usepackage[latin1]{inputenc}
\usepackage{color}
\usepackage{babel}
\usepackage{amsmath}
\usepackage{amsthm}
\usepackage{amssymb}
\usepackage{graphicx,subfigure}
\usepackage{esint}
\usepackage[unicode=true,
 bookmarks=false,
 breaklinks=false,pdfborder={0 0 1},backref=false,colorlinks=true]
 {hyperref}
\hypersetup{
 linkcolor=cyan,citecolor=magenta,filecolor=magenta,urlcolor=cyan,runcolor=cyan}
 \usepackage[normalem]{ulem}
 \usepackage{braket}

\makeatletter

\newtheorem{thm}{Theorem}

\newtheorem{cor}{Corollary}
\newtheorem{prop}{Proposition}
\newtheorem{lem}{Lemma}
\newtheorem{rem}{Remark}

\newtheorem{myproposition}{Proposition}

\newcommand{\bes}{\begin{subequations}}
\newcommand{\ees}{\end{subequations}}
\newcommand{\beq}{\begin{equation}}
\newcommand{\eeq}{\end{equation}}
\def\bea{\begin{subequations}\begin{align}}
\def\eea{\end{align}\end{subequations}}
\def\>{\rangle}
\def\<{\langle}

\newcommand{\ketb}[2]{|{#1}\>\!\<#2|}
\newcommand{\mc}{\mathcal}

\def\a{\alpha}
\def\b{\beta}
\def\g{\gamma}

\def\r{\rho}     

\def\o{\omega}

\newcommand{\ignore}[1]{}

\usepackage{times}

\newcommand{\hc}{\mathrm{h.c.}}

\DeclareMathOperator{\Tr}{Tr}

\newcommand{\1}{\leavevmode{\rm 1\ifmmode\mkern  -4.8mu\else\kern -.3em\fi I}}

\makeatother

\begin{document}
\title{Optimal Control for Closed and Open System Quantum Optimization
}
\author{Lorenzo Campos Venuti}
\affiliation{Department of Physics \& Astronomy, University of Southern California,
Los Angeles, California 90089, USA}
\affiliation{Center for Quantum Information Science \& Technology, University of
Southern California, Los Angeles, California 90089, USA}
\author{Domenico D'Alessandro}
\affiliation{Department of Mathematics, Iowa State University, Ames, Iowa 50014,
USA}
\author{Daniel A. Lidar}
\affiliation{Department of Physics \& Astronomy, University of Southern California,
Los Angeles, California 90089, USA}
\affiliation{Center for Quantum Information Science \& Technology, University of
Southern California, Los Angeles, California 90089, USA}
\affiliation{Department of Electrical \& Computer Engineering, University of Southern California,
Los Angeles, CA 90089, USA}
\affiliation{Department of Chemistry, University of Southern California, Los Angeles,
CA 90089, USA}
\date{\today}

\begin{abstract}
We provide a rigorous analysis of the quantum optimal control problem in the setting of a linear combination $s(t)B+(1-s(t))C$ of two noncommuting Hamiltonians $B$ and $C$. This includes both quantum annealing (QA) and the quantum approximate optimization algorithm (QAOA). The target is to minimize the energy of the final ``problem'' Hamiltonian $C$, for a time-dependent and bounded control schedule $s(t)\in [0,1]$ and $t\in \mc{I}:= [0,t_f]$. It was recently shown, in a purely closed system setting, that the optimal solution to this problem is a ``bang-anneal-bang'' schedule, with the bangs characterized by $s(t)= 0$ and $s(t)= 1$ in finite subintervals of $\mc{I}$, in particular $s(0)=0$ and $s(t_f)=1$, in contrast to the standard prescription $s(0)=1$ and $s(t_f)=0$ of quantum annealing. 

Here we extend this result to the open system setting, where the system is described by a density matrix rather than a pure state. This is the natural setting for experimental realizations of QA and QAOA. For finite-dimensional environments and without any approximations we identify sufficient conditions ensuring that either the bang-anneal, anneal-bang, or bang-anneal-bang schedules are optimal, and recover the optimality of $s(0)=0$ and $s(t_f)=1$. However, for infinite-dimensional environments and a system described by an 
adiabatic Redfield master equation we do not recover the bang-type optimal solution. In fact we can only identify conditions under which $s(t_f)=1$, and even this result is not recovered in the fully Markovian limit. 

The analysis, which we carry out entirely within the geometric framework of Pontryagin Maximum Principle, simplifies using the density matrix formulation compared to the state vector formulation. 
This analysis reveals that the bang-anneal-bang optimality result requires the assumption that the optimal control schedule is such that it is possible to lower the cost by increasing the total time $t_f$. A necessary condition for this is that $t_f$ is smaller than the critical time $t_c$ needed to reach the ground state of $C$ exactly. In previous work this condition was believed to also be sufficient, but we give a counterexample.

We derive a ``switching equation'' which describes the behavior of the optimal schedule switches between the two types of bangs and the anneals, and use it to identify the general features of the optimal control protocols. As an illustration of the theory, we analyze the simple example of a single spin-$1/2$, and prove that in this case the optimal solution in the closed system setting is the bang-bang schedule, switching midway from $s\equiv 0$ to $s\equiv1$.
\end{abstract}
\maketitle

\section{Introduction}
\label{sec:intro}

There is a great deal of interest in optimization algorithms that can be run on today's noisy intermediate scale quantum (NISQ) information processors~\cite{Preskill:2018aa}. Two prime examples are quantum annealing (QA)~\cite{kadowaki_quantum_1998} and the quantum approximate optimization algorithm (QAOA)~\cite{farhi2014quantum}. Both algorithms switch between two non-commuting Hamiltonians: a ``driver'' (or ``mixer'') $B$ and a ``target'' (or ``problem'') $C$. The latter encodes the solution to the optimization problem as its ground state. The two algorithms can be viewed as complementary: QA switches continuously while QAOA switches discretely; hence they are particularly well suited for analog and gate-model devices, respectively. In addition, both algorithms are related to the quantum adiabatic algorithm~\cite{Farhi:00}, which is guaranteed by the adiabatic theorem~\cite{Kato:50} to converge to the optimal solution in the limit of arbitrarily long evolution times~\cite{Jansen:07,lidar:102106,Ge:2015wo}. QA relaxes the strict adiabaticity condition while retaining continuity~\cite{kadowaki_quantum_1998}, and the adiabatic algorithm becomes an instance of QAOA when the continuous evolution is ``Trotterized'' (replaced by pulsed segments)~\cite[Sec.~VI]{farhi2014quantum}. There have been numerous studies of the two algorithms, including some that have compared them, with mixed results~\cite{bapat2018bang,zhou2018quantum,streif2019comparison,Pagano:2020wp}.

In essence, the question of which algorithm performs best  -- QA or QAOA -- boils down to an optimization of the switching schedule. Various results have already been established within the framework of the adiabatic algorithm, QA, or QAOA. For example, it is well known that the adiabatic algorithm can benefit from schedule optimization, even to an extent that can affect whether it provides a quantum speedup or not, as in the case of the Grover search problem~\cite{Roland:2002ul,RPL:10}. It has also been established that a variational approach can optimize the adiabatic switching schedule~\cite{PhysRevLett.103.080502}. Likewise, optimality results are known for QA~\cite{morita:125210,Galindo:2020ti} and QAOA~\cite{szegedy2019qaoa}. A natural question is whether one can jointly treat QA and QAOA under a single schedule optimization framework. The first such attempt was made by Yang \textit{et al.}~\cite{Yang:2017aa} using the framework of the Pontryagin
Maximum Principle (PMP) of optimal control~\cite{kipka_pontryagin_2015} (see also Refs.~\cite{Lin:2019wt,mbeng2019optimal,mbeng2019quantum}), whose conclusions favoring a strict QAOA-type schedule were later shown to be overly restrictive by Brady \textit{et al.}~\cite{brady_optimal_2021}, who showed that in general a hybrid discrete-continuous schedule is optimal. 

The results of Brady \textit{et al.}  were obtained in a closed system setting of purely unitary dynamics. 
Here, we generalize  the theory to the open system setting, and obtain their closed system results as a special case. We proceed to first provide the general background for the problem, after which we outline the structure of the rest of the paper.

\section{Background}
\label{sec:background}

The closed system setting involves a system evolving unitarily in a $d$-dimensional Hilbert space $\mathcal{H}$ subject to the Schr\"odinger equation:
\beq
\label{eq:Ham_sys}
\frac{d}{dt}|\psi(t)\rangle =-iH(t)|\psi(t)\rangle,\qquad|\psi(0)\rangle=|\psi_{0}\rangle\ .
\eeq
The protoypical quantum annealing problem concerns finding the optimal schedule $s(t)\in[0,1]$ for the time-dependent Hamiltonian given by\footnote{In the control literature the notation $u$ or $u(t)$ is used to denote the control function, rather than $s$ or $s(t)$. Here we choose to use the notation that is more familiar in the quantum computing community.}
\bes
\label{eq:Ham_sys2}
\begin{align}
H(t) &= s(t) B + (1-s(t)) C  \ , \quad t\in\mathcal{I} \\
 &= C +s(t)(B-C)\ .
 \label{eq:Ham_sys2-b}
\end{align}
\ees
The \emph{control interval} is $\mathcal{I} = [0,t_f]$.
Often the Hermitian operator $C$ is an Ising-type Hamiltonian of the form $\sum_{i=1}^n h_i \sigma^z_i + \sum_{i<j}^n J_{ij}\sigma^z_i\sigma^z_j$ (where $h_i$ and $J_{ij}$ are local longitudinal fields and couplings, respectively, and $\sigma^z_i$ is the Pauli matrix acting on the $i$'th qubit), and the Hermitian operator $B$ is a transverse field of the form $\sum_{i=1}^n \sigma^x_i$~\cite{kadowaki_quantum_1998}. For our purposes it only matters that $[B,C]\neq 0$.

The initial state $|\psi_{0}\rangle$ is assumed to be the ground state of $B$, and in both QA and QAOA the target state is the ground state of $C$. A relaxation of this, which we consider as the objective in the present work, is to minimize the expectation value of $C$ at a given final time $t_f$, i.e., 
\beq
J:=\langle\psi(t_f)|C|\psi(t_f)\rangle\ .
\label{eq:Brady-J}
\eeq 
Minimizing $J$ is equivalent to minimizing the energy of the $C$ Hamiltonian, and if the global minimum is found then this corresponds to finding the ground state of $C$ (i.e., solving the optimization problem defined by $\{h_i,J_{ij}\}$ when $C$ is in Ising form). 

It is known in quantum control theory (see, e.g., Ref.~\cite{dalessandro_introduction_2007})
that if $\hat{\mathfrak{L}}$ is the Lie algebra generated by $B$ and $C$ and $e^{\hat{\mathfrak{L}}}$ the corresponding Lie
group, assumed to be compact, the set of states reachable from $|\psi_{0}\rangle$
with free final time $t_f$ is 
\beq
\mathfrak{R}:=\{X|\psi_{0}\rangle\,|\,X\in e^{\hat{\mathfrak{L}}}\},
\label{eq:reach}
\eeq
so that the absolute minimum of the cost $J$ is 
\begin{equation}
J_{\min}=\min_{|\psi\rangle\in\mathfrak{R}}\langle\psi|C|\psi\rangle\ .
\label{absmin}
\end{equation}
If the dynamical Lie algebra $\hat{\mathfrak{L}}$ is the whole $su(d)$ then any state in the Hilbert space can be reached (starting from any other state), in particular the ground state of $C$, in which case the system is said to be \emph{controllable}~\cite{jurdjevic_control_1972}. 
However, requiring full controllability may be overly restrictive, as we only need to reach a particular state. The following is a simple generalization that provides a sufficient condition for reaching the ground state. 

\begin{myproposition}
\label{prop:0}
Suppose $[B,P_0]=[C,P_0]=0$ where $P_0$ is an orthogonal projector. The Hilbert space decomposes according to the block structure $\mathcal{H} = \mathrm{Ran}(P_0) \oplus \mathrm{Ker}(P_0)$. The Lie algebra generated by $B$ and $C$, $\hat{\mathfrak{L}}$, must have the same block structure. Suppose that,  according to this structure,  $\hat{\mathfrak{L}} = su(d_0)\oplus \hat{\mathfrak{L}}_1$ with unspecified $\hat{\mathfrak{L}}_1$ and $d_0=\dim(\mathrm{Ran}(P_0))$; then, if the initial state belongs to $\mathrm{Ran}(P_0)$, any state in $\mathrm{Ran}(P_0)$ can be reached (in finite time). 
\end{myproposition}

The proof is self-evident, since the full controllability result~\cite{jurdjevic_control_1972} is now applicable in $\mathrm{Ran}(P_0)$. This generalization can be applied, for example, in case both $B$ and $C$ commute with a third operator, say $M$, and one knows to which sector of $M$ the ground state of $C$ belongs; see Appendix \ref{app:reachability} for an example. In any case, it is clear that something must be assumed in order to guarantee the reachability of the ground state of $C$. Clearly, a necessary condition is $[B,C]\neq 0$, but even when $[B,C]\neq 0$ it is easy to come up with examples where the ground state of $C$ cannot be reached; see Appendix \ref{app:reachability}. In the following we will tacitly assume that conditions are such that the ground state of $C$ can be reached.  

Brady \textit{et al.}~\cite{brady_optimal_2021} used optimal control methods to prove that for the cost as defined as in Eq.~\eqref{eq:Brady-J}, the optimal schedule is one where at the beginning and end of the control interval $s \equiv 0$ and $s \equiv 1$, respectively.\footnote{We use the notation $f \equiv c$ to mean that the function f is ``identically'' equal to $c$ in some interval $\Theta$, i.e., $f \equiv c$ is equivalent to $f(t) = c$ for $t\in\Theta$.} From a quantum annealing perspective this might appear as a counterintuitive result, since it means that rather than the usual ``forward'' formulation of quantum annealing~\cite{Albash-Lidar:RMP,Hauke:2019aa}, where one interpolates smoothly from $H(0) = B$ to $H(t_f) = C$, 
the optimal protocol starts from the system being in the ground state of $B$ but the initial Hamiltonian is $C$, and the final Hamiltonian is not $C$ but rather $B$. 
The result, however, can be understood by noting that in the adiabatic approach, one interpolates so slowly from $H(0) = B$ to $H(t_f) = C$ that the system always remains in the ground state. Instead, in optimal control, we optimize over the set of possible states obtained by applying either $C$ or $B$ to the initial state, in a continuous fashion. In this sense, applying $B$ at the beginning is a waste of time as it does not change the initial state. Applying $C$ at the end, when the system is supposed to be close to the ground state of $C$, is similarly wasteful. This relaxation of the approach of strict adiabaticity is in line with other alternatives, such as shortcuts to adiabaticity~\cite{Campo:2013ix,PhysRevA.95.012309} and diabatic quantum annealing~\cite{crosson2020prospects}.

More precisely, Ref.~\cite{brady_optimal_2021} showed, provided that a certain condition 
holds (see below), that the optimal control function starts (ends) with $s\equiv 0 $  ($ s\equiv1$) in an interval of positive measure after $t=0$ (before $t_f$). Elsewhere the optimal control $s(t)$ is ``singular'', except for possible interruptions by a sequence of ``bang'' controls, where $s\equiv0$ or
$s\equiv1$. 
In control theory  a ``singular'' interval or arc, is an interval of time where  the PMP control Hamiltonian in Eq.~\eqref{HamilP} below does not depend on the control $s$. The remaining ``nonsingular'' arcs give rise to the ``bang'' controls. In the numerical simulations of Ref.~\cite{brady_optimal_2021}, the control appeared to be continuous (even smooth) on such singular arcs. Hence the term  ``anneal'' was used in lieu of ``singular'', with the intention of stressing the continuous (or possibly even smooth) nature of the control on the singular arcs. 
They suggestively called the resulting optimal control a ``bang-anneal-bang'' protocol.
At present, a rigorous proof that the control function is continuous (let alone smooth) on singular arcs is lacking,
and there is some risk of confusion in interpreting the singular arcs as always being continuous, or even differentiable as is typically assumed in QA and adiabatic quantum computing~\cite{Jansen:07,Albash-Lidar:RMP}. Nonetheless, while keeping these caveats in mind, we shall adopt the same (numerically supported) terminology as Ref.~\cite{brady_optimal_2021}, and use ``continuous (or anneal) $=$ singular'' as well as ``bang $=$ nonsingular'' interchangeably. 

Here, we consider the \emph{open system} version of the same optimal control problem. We reformulate the problem in terms of the density matrix $\rho$, whose dynamics is described by the
following, rather general master equation:\footnote{The form we have assumed is called a time-convolutionless master equation. The most general master equation is in Nakajima-Zwanzig form and includes a memory kernel superoperator $\mathcal{K}(t,t')$ acting jointly on the system and the environment $E$, such that (for a factorized initial condition) $\dot{\rho} = \Tr_E \int_0^t \mathcal{K}(t,t')\rho(t')\otimes\rho_E dt'$, with $\rho_E$ a fixed environment state and $\Tr_E$ denoting the partial trace over the environment~\cite{Breuer:book}.}
\begin{equation}
\dot{\rho}=\mathcal{L}\rho,\qquad\rho(0)=\rho_{0}\ ,
\label{eq:general_ME}
\end{equation}
where the Liouvillian $\mathcal{L}$ depends linearly on the control $s$ (and the controlled  operators $B,C$). Note that the Liouvillian is not explicitly time-dependent (i.e.~$\partial_t \mathcal{L}=0, \ \forall t\in \mathcal{I}$) and depends on time only through the control schedule $s(t)$. This is an important requirement that will play a
crucial role in our ability to apply the Pontryagin principle in the form we need,
as we discuss in more detail below. Furthermore, to be physically meaningful,  $\mathcal{L}$ must preserve hermiticity, i.e., $[\mathcal{L}(X)]^\dagger = \mathcal{L}(X^\dagger), \ \forall X$.  Instead of Eq.~\eqref{eq:Brady-J}, the cost $J$ takes
the form
\begin{equation}
J:=\Tr\left[C\rho(t_f)\right] =\langle C,\rho(t_f)\rangle\ ,
\label{nuovocosto}
\end{equation}
where we used the Hilbert-Schmidt scalar product $\langle X,Y\rangle := \Tr(X^\dagger Y)$ for operators $X, Y$. 
We shall see that a description and treatment of the optimal control problem in the setting of the density matrix is not only more general but also more elegant since the cost $J$ is linear in the state rather than quadratic, as in Eq.~\eqref{eq:Brady-J}. Moreover, we obtain the closed system result as a special case. Unlike Ref.~\cite{brady_optimal_2021}, which used a mixture of the PMP and a variational (Lagrange multiplier type of) argument, we use only the PMP, which significantly simplifies the proof.

The rest of this paper is organized as follows. In Sec.~\ref{sec:generalME},
we apply general results from optimal control theory and the necessary conditions of the PMP to the
problem of minimizing $J$ [Eq.~\eqref{nuovocosto}] for a given final
time $t_f$ and the general dynamical system of the form of Eq.~\eqref{eq:general_ME}.
In Sec.~\ref{OCP} we specialize to the case of closed systems, which are described
by the von Neumann equation. In particular, we confirm but also sharpen the results of Ref.~\cite{brady_optimal_2021}. We also analyze in depth the optimal control problem of a single spin-$1/2$, and prove that the optimal schedule is of the bang-bang type. In Sec.~\ref{OSC} we consider the case
of open systems. This includes both the most general 
case of a reduced description of quantum system obtained by tracing out the environment it is coupled to, and the case where the open
quantum system is described by adiabatic master equations, both non-Markovian and Markovian. In Sec.~\ref{Switchings} we derive a ``switching equation,'' which allows us to provide a general characterization of the switches between non-singular and singular arcs, and derive conditions for the presence or absence of singular
arcs. We also give a heuristic derivation of the shortening of the length of the arcs between two switches with increasing system size. We conclude in Sec.~\ref{sec:conc}.
In a series of appendices we provide additional background on optimal control theory and technical details and proofs of various results from the main text.

\section{Statement of the Pontryagin Maximum Principle}
\label{sec:generalME}

We state the PMP as it applies to our problem of interest (see Appendices~\ref{app:A} and~\ref{app:real}):

\begin{thm}
\label{adapt0} 
Assume that $\rho^*$ and $s^*$ are, respectively, an optimal state and control pair for the problem defined by Eqs.~\eqref{eq:general_ME} and~\eqref{nuovocosto} for a fixed final time $t_f$.\footnote{We also use an asterisk to denote complex conjugation; the meaning will always be clear by context.} Then there exists a nonzero $n\times n$ Hermitian time-dependent matrix $p=p(t)$ called the {\em co-state} that satisfies\footnote{$\mathcal{L}^{\dag}$ in Eq.~\eqref{LiouvK1} indicates the Hilbert-Schmidt adjoint of $\mathcal{L}$, defined via $\langle \mathcal{L}^\dag(A),B\rangle := \langle A,\mathcal{L}(B)\rangle = \Tr[A^\dag \mathcal{L}(B)], \ \forall A,B$; see Appendix~\ref{app:real}.}
\begin{equation}
\dot{p}=-\mathcal{L}^{\dag}p\ ,
\label{LiouvK1}
\end{equation}
with the final condition 
\begin{equation}
p(t_f)=-C\ .
\label{3plusB}
\end{equation}
Furthermore, define the \emph{PMP control Hamiltonian} function 
\begin{equation}
\mathbb{H}\left(p,\rho,s\right)
:=\langle p,\mathcal{L}\rho\rangle\ .
\label{HamilP}
\end{equation}
We then have the \emph{maximum principle}: 
\begin{equation}
\mathbb{H}\left(p(t),\rho^*(t),s^*(t)\right)=\max_{v\in[0,1]}\mathbb{H}\left(p(t),\rho^*(t),v\right)\ ,
\label{MaxCon}
\end{equation}
and there exists a real 
constant $\lambda$ such that
\begin{equation}
\mathbb{H}\left(p(t),\rho^*(t),s^*(t)\right)\equiv \lambda\ .
\label{forlanda}
\end{equation}
\end{thm}

A few remarks are in order.
\begin{itemize}
\item The PMP control Hamiltonian \emph{function}~\eqref{HamilP} is, of course, different from the Hamiltonian \emph{operator}~\eqref{eq:Ham_sys2} generating the dynamics. 
\item Since $p$ and $\rho$ are Hermitian and $\mathcal{L}$
is Hermiticity-preserving [$\left[\mathcal{L}\left(X\right)\right]^{\dagger}=\mathcal{L}\left(X^{\dagger}\right)$ $\forall X$], ``expectation values'' of the form $\langle p,\mathcal{L}\rho\rangle$
are real, and hence so is the PMP control Hamiltonian~\eqref{HamilP}.  
\item The condition $\partial_t\mc{L}=0$ $\forall t \in \mathcal{I}$ must be satisfied and is implicit in Eq.~\eqref{eq:general_ME}. In other words, $\mc{L}$ may not depend explicitly on time. Without this condition Eq.~\eqref{forlanda} does not hold with a constant $\lambda$.
\item It is worth highlighting that at the final time the co-state becomes the (negative of the) target Hamiltonian [Eq.~\eqref{3plusB}], a fact we use repeatedly in our applications of the Theorem~\ref{adapt0} below.
\item As discussed after Theorem~\ref{case1} in Appendix~\ref{app:A}, if the optimal trajectory is such that the constraint of the final time $t_f$ is ``active'', i.e., a small perturbation $t_f+\delta$
allows us to decrease the cost, then $\lambda>0$ in Eq.~\eqref{forlanda}. This is an important point that will be further discussed in the next section. 
\item Given that the PMP is formulated in terms of real-valued quantities in the optimal control literature (see Appendix~\ref{app:A}), one must first transform the relevant equations into real-valued ones. This can easily be done since the space of $n\times n$ Hermitian matrices is isomorphic to the space of $n^2$ real variables via coordinatization. We discuss this in Appendix~\ref{app:real}.
\item Since $p(t)$ and $\rho(t)$ are solution of differential equations, they are continuous function of time. This implies that expressions of the form $\langle p, \mc{K}\rho\rangle$ with the superoperator $\mc{K}$ independent of time (both explicitly and implicitly), are  continuous functions of $t$, a fact which we repeatedly and implicitly use below.
\end{itemize}

\section{The closed system case}
\label{OCP}

We first consider the closed system case. Let us define the superoperator
\beq
\mathcal{K}_{X}:=-i\left[X,\bullet\right]\ .
\label{eq:K_X}
\eeq 
Note that $\mathcal{K}_{X}$
is linear with respect to $X$. For Hermitian $X$, $\mathcal{K}_{X}$
is anti-Hermitian (see Appendix~\ref{app:calcs}):
\beq
\mathcal{K}_{X}^{\dag}=-\mathcal{K}_{X^{\dagger}}=-\mathcal{K}_{X}\ .
\label{eq:KXanti}
\eeq 
The von Neumann equation corresponding to Eq.~\eqref{eq:Ham_sys} is
\begin{equation}
\dot{\rho}=\mathcal{K}_{C}\rho+s(t)\left(\mathcal{K}_{B}\rho-\mathcal{K}_{C}\rho\right),\quad\rho(0)=\rho_{0} \ ,
\label{eq:vonNeum}
\end{equation}
where henceforth we denote the initial and final conditions of operators $X$ by $X(0):=X_0$ and $X(t_f):=X_f$, respectively. 
I.e., one has Eq.~\eqref{eq:general_ME} with 
\beq
\mathcal{L}=\mathcal{K}_{C}+s(t)\mathcal{K}_{B-C}\ .
\label{eq:L-closed}
\eeq
Since in this case $\mathcal{L}^{\dag}=-\mathcal{L}$, Eq.~\eqref{LiouvK1} tells us that the co-state matrix $p$ satisfies the same equation as $\rho$: 
\begin{equation}
\dot{p}=\mathcal{K}_{C}p+s(t)\mathcal{K}_{B-C}p\ ,
\label{Eq:co_vonNeum}
\end{equation}
but with the final condition~\eqref{3plusB}.
The PMP control Hamiltonian reads 
\begin{equation}
\mathbb{H}=\langle p,\mathcal{K}_{C}\rho\rangle+s(t)\langle p,\mathcal{K}_{B-C}\rho\rangle\ .
\label{HamilP-1}
\end{equation}
 
 \subsection{The ``bang-anneal-bang'' protocol is optimal}
 
Applying Theorem~\ref{adapt0} to the anti-Hermitian superoperator $\mathcal{L}$ of Eq.~\eqref{eq:L-closed}, we obtain the following extension of the result of Ref.~\cite{brady_optimal_2021} to the density matrix setting:
\begin{thm}
\label{Bradygen}
(i) Assume $s^*\in[0,1]$ is the optimal control in an
interval $[0,t_f]$ minimizing the cost~\eqref{nuovocosto} for
Eq.~\eqref{eq:vonNeum}. Then there exists a nonzero
Hermitian matrix solution of Eq.~\eqref{Eq:co_vonNeum} with terminal
condition~\eqref{3plusB} such that $s^*\equiv0$ on intervals where
$\langle p,\mathcal{K}_{B-C}\rho\rangle < 0$, and 
$s^*\equiv1$ on intervals where $\langle p,\mathcal{K}_{B-C}\rho\rangle > 0$.
On all other intervals $\langle p,\mathcal{K}_{B-C}\rho\rangle\equiv0$
(these are called singular arcs). 

(ii) Assume furthermore that the constraint on the final time $t_f$ is active (so that $\lambda>0$). Then $s^*(t)=1$ for $t\in(t_f-\epsilon,t_f]$ for some $\epsilon>0$. Moreover, if the initial condition $\rho_0$ commutes with the driver Hamiltonian $B$, i.e., $\mathcal{K}_{B}\rho_0=0$, one also has  $s^*(t)=0$ for $t\in[0,\epsilon')$ for some $\epsilon'>0$. 
\end{thm}

Before proving this theorem we offer a few remarks. 
\begin{itemize}
\item Part (i) implies that the optimal control is, in general, an alternation of ``bang'' (nonsingular)
arcs and ``anneal'' (singular) arcs where $\langle p,\mathcal{K}_{B}\rho\rangle=\langle p,\mathcal{K}_{C}\rho\rangle$. Using the PMP, this is an immediate consequence of the fact that the control enters linearly
in the equation and it is coupled to the superoperator $\mathcal{K}_{B-C}$. \emph{The latter is what is ``special'' about the quantum annealing problem.}
\item Part (ii) implies that under the assumption of an active time constraint and  for a particular initial condition, the optimal control starts and ends with nonsingular arcs. In particular,
it starts with an arc $s\equiv0$ and ends with an arc $s\equiv1$.
\item In practice, whether there are additional nonsingular arcs in the middle is problem dependent, and there is numerical evidence that such optimal scenarios do indeed exist~\cite{brady_optimal_2021}, but such nonsingular arcs do not exist in the single qubit example discussed in Subsection \ref{esempio}. 
\item Assuming that $\mathcal{K}_{B}\rho_0=0$ one can prove $\lambda \ge0$  (see also Appendix~\ref{app:A}, Proposition~\ref{landamagz}), but the condition $\lambda >0$ is more subtle and it must be assumed independently. We return to this point in the next subsection.
\end{itemize}

\begin{proof}
Part (i): Eq.~\eqref{HamilP-1} states that the PMP control Hamiltonian $\mathbb{H}$ depends on the control only via the term $s(t)\langle p,\mathcal{K}_{B-C}\rho\rangle$. If $\langle p,\mathcal{K}_{B-C}\rho\rangle <0$, then to maximize this term as per Eq.~\eqref{MaxCon} subject to the constraint that $s(t)\in [0,1]$, clearly we must set $s^\ast\equiv 0$. Likewise, if  $\langle p,\mathcal{K}_{B-C}\rho\rangle > 0$, then to maximize this term subject to the same constraint requires $s^\ast\equiv 1$. This is the case of nonsingular arcs. Conversely, if $\langle p,\mathcal{K}_{B-C}\rho\rangle \equiv 0$ (a singular arc), then we cannot conclude anything about the control from the PMP.

Part (ii): 
To investigate the form of the control at the end
of the control interval $[0,t_f]$, consider Eq.~\eqref{forlanda} with $\lambda>0$. 
Using Eq.~\eqref{3plusB} we have $\langle p_f,\mathcal{K}_{C}\rho_f\rangle=\langle\mathcal{K}_{C}^{\dag}p_f,\rho_f\rangle=\langle\mathcal{K}_{C}C,\rho_f\rangle=0$.
This means that  $\mathbb{H}(t_f)=s(t_f) \langle p_f,\mathcal{K}_{B}\rho_f\rangle=\lambda>0$, which in turn, since $s\in[0,1]$,  implies that $\langle p_f,\mathcal{K}_{B}\rho_f\rangle>0$. By continuity there must exist an interval $(t_f-\epsilon,t_f]$ (for some $\epsilon>0$) such that $\langle p (t),\mathcal{K}_{B-C}\rho (t)\rangle>0$ for $t\in(t_f-\epsilon,t_f]$, and in this interval we must have $s^\ast(t)=1$ by (i).

The argument for the initial time is similar but instead of Eq.~\eqref{3plusB} it uses the extra assumption $\mathcal{K}_{B}\rho_{0}=0$. Let us evaluate the control Hamiltonian at $t=0$. Because of the assumption $\mathcal{K}_{B}\rho_0=0$ we have 
$\mathbb{H}(t=0)=(1-s(0))\langle p_0,\mathcal{K}_{C}\rho_0\rangle=\lambda>0$. Since $s\in[0,1]$ this implies that $\langle p_0,\mathcal{K}_{C}\rho_0\rangle>0$ and $\langle p_0,\mathcal{K}_{-C}\rho_0\rangle<0$. By continuity there must exist an $\epsilon'>0$ such that, for $t\in[0,\epsilon')$,  $\langle p(t),\mathcal{K}_{B-C}\rho(t)\rangle <0$ and in this interval we must have $s^\ast(t)=0$ by (i).
\end{proof}

\subsection{The active constraint assumption and a sharpening of the results of Ref.~\cite{brady_optimal_2021}}
\label{sec:active-constraint}

The condition $\lambda>0$ (that is, an active constraint on
the final time $t_f$) requires some extra discussion. It is a known
fact in the geometric theory of quantum control systems, and it follows
as an application of general results on control systems on Lie groups
(see, e.g., Ref.~\cite[Th.~7.2]{jurdjevic_control_1972}), that
there exists a \emph{critical time} $t_{c}$ such that, the set
of states reachable at time $t$, coincides for every $t\geq t_{c}$.
In other words, the reachable set does not grow past a certain time
$t_{c}$. Therefore, for every $t_f\geq t_{c}$ the time constraint
is never active. The minimum time $t_{\min}$ to reach the ground state of $C$ is $\leq t_c$. If the final 
time $t_f$ is greater than or equal to $t_{\min}$, then again the time constraint can never be active. In order to avoid this situation, 
it was claimed in Ref.~\cite{brady_optimal_2021} that having $t_f<t_{\min}$,  is sufficient for having $\lambda>0$ in Eq.~\eqref{forlanda}. Their argument only uses $[\rho_0,B]=0$. However, in Sec.~\ref{esempio} below we give an example satisfying this assumption for which $\lambda =0$ for arbitrarily small $t_f$. Thus, the assumption $t_f<t_{\min}$ 
is certainly necessary for $\lambda>0$ but is in fact not sufficient. Rather, $\lambda>0$ is
a feature of the optimal trajectory rather than of the problem itself.  This can be explained more easily geometrically, as we now do. 

The optimal cost at time $t_f$ is the minimum of
a continuous function on the reachable set of states $\mathfrak{R}_{t_f}$ (see Appendix~\ref{app:A}).
It is also known, under conditions that apply in our case, that the reachable set $\mathfrak{R}_{t_f}$
varies continuously with $t_f$~\cite{ayala_about_2017}. We can map the space of Hermitian
matrices $\rho$ diffeomorphically to $\mathbb{R}^{n^{2}}$ (see Appendix~\ref{app:real}) and consider
its reachable set there. Since the cost function~\eqref{nuovocosto}
is linear on this set, the minimum occurs on the boundary. Therefore,
the optimal trajectory is a curve starting from the initial condition
$\rho_{0}$ and ending on the boundary of $\mathfrak{R}_{t_f}$.
At the endpoint, the trajectory will have a tangent vector which indicates
its future direction. Now $\lambda>0$ if, going (infinitesimally)
in that direction combined with an increase in the size of the reachable
set $\mathfrak{R}_{t_f+\epsilon}$ for some small $\epsilon$, will result
in a reduced cost, and this is what we mean by the time
constraint being active. If $t_f$ is such that the reachable set does
not increase at $t_f$, for instance if $t_f\geq t_{c}$, then
clearly this is not possible and we must have $\lambda=0$. However,
it is also possible that the reachable set increases but not in a
way to (strictly) decrease the cost, in particular the portion of
the boundary where we landed might not move at all, or it might move
but not in a direction that decreases the cost. This geometric discussion is illustrated with figures in  Appendix~\ref{app:geomfig}. 

The  phenomenon that the optimal cost does not decrease with an increasing final time $t_f$ may occur
even though $t_f$ is arbitrarily small. Let us denote by $J_{\min}(t_f)$
the minimum cost~\eqref{nuovocosto} as a function of $t_f$. The
example we provide below (Sec.~\ref{esempio}) shows that, even assuming
$[\rho_{0},B]=0$, we can have $J_{\min}(t_f)=J_{\min}(0)$
for $t_f\in[0,\epsilon)$ and some $\epsilon>0$, that is,
the cost cannot be lowered for some time, independently of the control.
However, under the additional assumption that $\rho_{0}$ is the \emph{nondegenerate
ground state of $B$} 
this does not happen, and we have the following theorem which we prove
in Appendix~\ref{sec:proof-thm3}:

\begin{thm}
\label{Active} 
Assume that $\rho_{0}$ in Eq.~\eqref{eq:general_ME} is the nondegenerate ground state of $B$. Then there exists an $\epsilon>0$ such that, for every $t_f\in(0,\epsilon)$, $J_{\min}(t_f)<J_{\min}(0)=\Tr(C\rho_{0})$. 
\end{thm}

In other words, if we start from the nondegenerate ground state of $B$ we can always decrease the cost for sufficiently small $t_f$. Note that this, however, does not prove that $\lambda >0$. As we have explained, the condition $\lambda>0$ is a condition about the optimal trajectory, and 
it is an open problem to find sufficient conditions such that every optimal trajectory satisfies the $\lambda > 0$ requirement for sufficiently small $t_f$.

\subsection{Example: optimal control of a spin-$1/2$ particle}
\label{esempio}

We now give an example showing that without the assumption that $\rho_{0}$ is the nondegenerate ground state of $B$, the cost~\eqref{nuovocosto} cannot be lowered even for arbitrarily small $t_f$'s.

Consider a spin-$1/2$ particle (qubit) in a magnetic field. The model
is given by Eq.~\eqref{eq:vonNeum} with $C=\frac{1}{2}\sigma^z$
and $B=\frac{1}{2}\sigma^x$. As an orthonormal, Hermitian operator basis we choose $F_i = \frac{1}{\sqrt{2}}\sigma_i$, where we denote the standard Pauli matrices $\sigma^x\equiv\sigma_1$ etc., i.e.:
\begin{equation}
\sigma_{1}=\left(\begin{array}{cc}
0 & 1\\
1 & 0
\end{array}\right),\quad\sigma_{2}=\left(\begin{array}{cc}
0 & -i\\
i & 0
\end{array}\right),\quad\sigma_{3}=\left(\begin{array}{cc}
1 & 0\\
0 & -1
\end{array}\right)\ .
\label{PauliMat}
\end{equation}
They satisfy the su$(2)$ commutation relations 
\begin{equation}
\left[\sigma_1,\sigma_2\right]=2i\sigma_3,\quad\left[\sigma_3,\sigma_1\right]=2i\sigma_2,\quad\left[\sigma_2,\sigma_3\right]=2i\sigma_1\ .
\label{commurel}
\end{equation}
We parametrize the density matrix as $\rho=\frac{1}{2}\left(\1+\boldsymbol{v}\cdot\boldsymbol{\sigma}\right)$,
where $\boldsymbol{v}\in \mathbb{R}^{3}$ is the Bloch vector ($\|\boldsymbol{v}\|\leq 1$) and $\boldsymbol{\sigma} = (\sigma_1,\sigma_2,\sigma_3)^T$. The Bloch vector satisfies Eq.~\eqref{eq:vonNeum} where, using $\boldsymbol{\mathcal{K}}_{ij} = \Tr[F_i \boldsymbol{\mathcal{K}}(F_j)]$ (see Appendix~\ref{app:real}), we
have
\begin{equation}
\label{adcadb}
\boldsymbol{\mathcal{K}}_{B}=\left(\begin{array}{ccc}
0 & 0 & 0\\
0 & 0 & -1\\
0 & 1 & 0
\end{array}\right),\quad\boldsymbol{\mathcal{K}}_{C}=\left(\begin{array}{ccc}
0 & -1 & 0\\
1 & 0 & 0\\
0 & 0 & 0
\end{array}\right)\ .
\end{equation}
Equivalently, the dynamics are given by the Bloch equation
\bes
\begin{align}
\label{eq:Bloch-eq} 
\dot{\boldsymbol{v}} & =\boldsymbol{\mathcal{M}}(s)\boldsymbol{v}
\\
\boldsymbol{\mathcal{M}}(s) & =\left(\begin{array}{ccc}
0 & -(1-s) & 0\\
1-s & 0 & -s\\
0 & s & 0
\end{array}\right)\ .
\label{eq:M_Bloch}
\end{align}
\ees
Geometrically, $\boldsymbol{\mathcal{K}}_{C}$ is the infinitesimal
generator of a counterclockwise rotation about the $v_{3}$ axis, while $\boldsymbol{\mathcal{K}}_{B}$
is the infinitesimal generator of a counterclockwise rotation about the $v_{1}$
axis. For $s\in(0,1)$, $\boldsymbol{\mathcal{M}}(s)=(1-s)\boldsymbol{\mathcal{K}}_{C}+s\boldsymbol{\mathcal{K}}_{B}$
generates a counterclockwise rotation about an intermediate axis in the $(v_{1},v_{3})$
plane. The cost~\eqref{nuovocosto} in this case becomes $J = \frac{1}{2}v_3$, i.e., it corresponds to the minimization of the $v_{3}$ component. Furthermore, let us
assume for simplicity that the initial state $\rho_{0}$ is pure ($\|\boldsymbol{v}\| = 1$). There are only two such states compatible with the condition $[B,\rho_{0}]=0$ (equivalently: $\boldsymbol{\mathcal{K}}_{B}\boldsymbol{v}_0=\boldsymbol{0}$): the $\sigma_1$ eigenstates, i.e., $\boldsymbol{v}_{0}=(\pm1,0,0)^{T}$. 

Now, if the initial state is $\boldsymbol{v}_{0}=(1,0,0)^{T}$, i.e., the excited state of $B$, then for sufficiently small $t$ we have
$v_{3}(t)\geq0$ independently of the control $s\in[0,1]$ (see Appendix~\ref{app:calcs}).
Therefore,
an optimal control in $[0,t_f]$ for $t_f$ small will be $s\equiv0$
(which will keep the value of $v_{3}$ at zero). The value of the
minimum cost $J_{\min}(t_f)$ is equal to $J(0)$ for any
arbitrarily small $t_f$. The constraint on the final time is not
active here, \emph{even for arbitrarily small $t_f$}. As a
consequence, in this case we cannot draw the conclusions of Theorem
\ref{Bradygen} following from the assumption $\lambda>0$. On the
other hand, for $\boldsymbol{v}_{0}=(-1,0,0)^{T}$, which corresponds to the (nondegenerate)
\emph{ground state}, with sufficiently small $t_f$ we can lower the cost
according to Theorem~\ref{Active}. We prove in Appendix~\ref{low} that the optimal control in this case is a simple bang-bang protocol:

\begin{thm}
\label{conclusione} 
The  optimal control for the system of one spin-$1/2$ particle considered above,
starting from the ground
state and minimizing the cost $\Tr(\sigma_{3}\rho)$ in time $t_f<\pi$, is the sequence $s^\ast\equiv0$ for time $\frac{t_f}{2}$ followed by $s^\ast\equiv1$ for time $\frac{t_f}{2}$ (see Fig.~\ref{fig:optimal}). 
\end{thm}
Here $t_c=\pi$, i.e., if $t_f \ge \pi$ one trivially finds the ground state exactly (by a $\pi/2$ rotation from the $-1$ eigenstate of $\sigma^x$ to the $-1$ eigenstate of $\sigma^y$, followed by another $\pi/2$ rotation to the $-1$ eigenstate of $\sigma^z$) and one cannot do better by increasing $t_f$. This optimal bang-bang schedule result for a single spin-$1/2$ joins previous such results for systems as diverse as pairs of one-dimensional quasicondensates~\cite{Rahmani:2013vq} or ``gmon'' qubits~\cite{Bao:2018uh}, as well as braiding of Majorana zero modes~\cite{Karzig:2015un}.

\begin{figure}
\begin{centering}
\includegraphics[width=7cm]{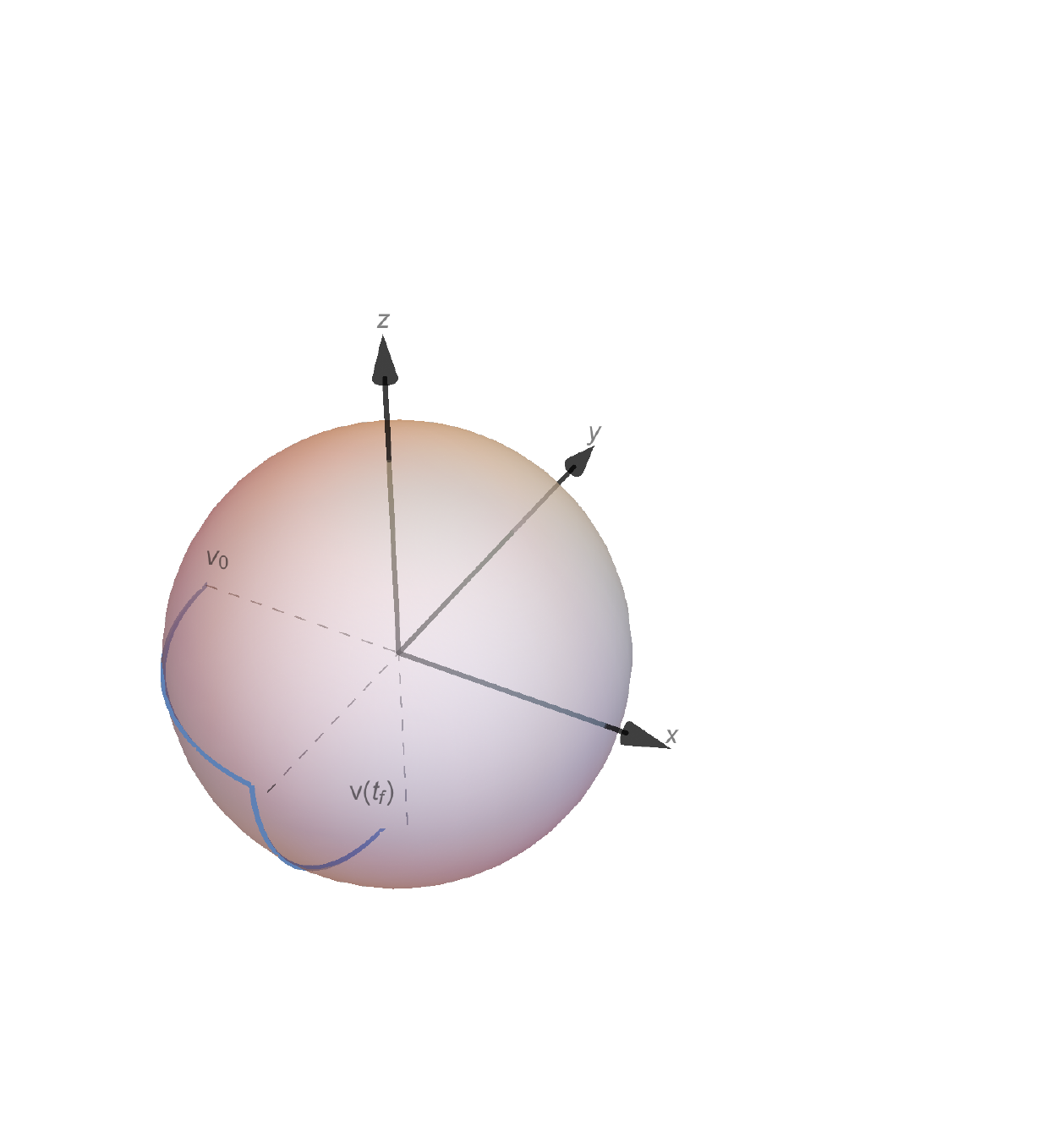}
\par\end{centering}
\caption{(Color online) Single qubit case. Optimal trajectory on the Bloch sphere with initial condition $\boldsymbol{v}_0 = (-1,0,0)^T$. Here $t_f=0.95\pi$ so the ground state of $C$, corresponding to the point $(0,0,-1)$,  is not reached exactly.}
 \label{fig:optimal}
\end{figure}

\section{The open system case}
\label{OSC}

In this section we generalize the results for closed systems to the open system setting. We consider two different approaches: an approximation-free treatment of a system + environment where both are finite-dimensional, and a master equation approach subject to a Markovian approximation, which applies for infinite-dimensional environments~\cite{Davies:76,alicki_quantum_2007,Breuer:book,rivas_open_2012,Lidar:2019aa}.  We show that under a number of additional assumptions, we can (partially) recover the results from the closed system setting, but that the bangs characterizing the latter are not a particularly robust feature in the open system setting.

\subsection{Optimal control for the Liouville-von Neumann equation}
\label{OCLVN}

One approach for extending the closed system results of the previous section to open systems is to consider the \emph{full dynamics} of a jointly evolving system + environment. In this case
$\rho$ in Eq.~\eqref{eq:general_ME} is the density matrix of the system
and the environment, with an initial condition which is usually taken to be of the
factorized form $\rho(0):=\rho_{0}\otimes\rho_{E}$, where now $\r_0$ refers to the initial state of the system only. The Liouville-von Neumann equation is [extending Eq.~\eqref{eq:vonNeum}]:
\begin{equation}
\dot{\rho}=\mathcal{K}_{H_{\mathrm{tot}}}\rho\ , \quad \rho(0)=\rho_{0}\otimes\rho_{E}\ ,
\label{eq:open}
\end{equation}
where $\mathcal{K}$ is defined in Eq.~\eqref{eq:K_X}, with the total Hamiltonian 
\begin{equation}
H_{\mathrm{tot}}=H_S\otimes \1_E+H_{I}+\1_S\otimes H_{E}\ .
\label{eq:H_tot}
\end{equation}
Here $H_{I}$ is the interaction between the system and
environment, $H_{E}$
generates the dynamics of the environment, while the system Hamiltonian,
as before, contains the controllable part:
\beq
H_S(t)=C+s(t)\left(B-C\right)\ .
\label{eq:H_S}
\eeq
Finally the cost is given by 
\begin{equation}
J=\Tr[\left(C\otimes\1_{E}\right)\rho]\ .
\label{eq:cost_open}
\end{equation}
Theorem~\ref{adapt0} holds with $\mathcal{L}=\mathcal{K}_{H_{\mathrm{tot}}}$
and the PMP control Hamiltonian has the form 
\begin{equation}
\mathbb{H}=\langle p,\mathcal{K}_{H_S\otimes\1_E}\rho\rangle+\langle p,\mathcal{K}_{H_{I}}\rho\rangle+\langle p,\mathcal{K}_{\1_S\otimes H_{\mathrm{E}}}\rho\rangle\ .
\label{eq:Ham_open}
\end{equation}
The treatment of Sec.~\ref{OCP} applies, \textit{mutatis mutandis}. In
particular, condition~\eqref{3plusB} is replaced by 
\begin{equation}
p(t_f)=-C\otimes\1_E\ .
\label{eq:pt_H_open}
\end{equation}
Remarkably, no additional modifications of the statement of the PMP Theorem~\ref{adapt0} are needed. Moreover, it is clear from Eqs.~\eqref{eq:H_S} and~\eqref{eq:pt_H_open} that once again the control enters $\mathbb{H}$ only via the term $s(t)\langle p,\mathcal{K}_{(B-C)\otimes\1_E}\rho\rangle$, so that the proof of Part (i) of Theorem~\ref{Bradygen} applies without any change. This shows that:
\begin{cor}
In the general open system setting of Eq.~\eqref{eq:open} with the cost~\eqref{eq:cost_open}, the optimal control $s(t)$ is an alternation of bang arcs where $s\equiv0$ (when $\langle p,\mathcal{K}_{(B-C)\otimes\1}\rho\rangle<0$) or    $s\equiv1$ (when $\langle p,\mathcal{K}_{(B-C)\otimes\1}\rho\rangle>0$),
and singular arcs where $\langle p,\mathcal{K}_{C\otimes\1}\rho\rangle\equiv\langle p,\mathcal{K}_{B\otimes\1}\rho\rangle$.
\label{control_general_open}
\end{cor}

Let us consider the generalization of Part (ii) of Theorem~\ref{Bradygen}, which
addresses the characterization of the control function at the beginning
and at the end. We first consider the final arc. We have the following:
\begin{thm}
\label{thm:open_end} 
Assume $s^*\in[0,1]$ is the optimal control in an
interval $[0,t_f]$ minimizing the cost~\eqref{eq:cost_open} for
Eq.~\eqref{eq:open}. 
Assume furthermore that $\left[H_{I},C\otimes\1_{E}\right]=0$
and that the final time constraint is active, i.e., $\lambda>0$. Then 
$s^*(t)=1$ for $t\in(t_f-\epsilon,t_f]$ for some $\epsilon>0$.
\end{thm}

Note that the assumption $\left[H_{I},C\otimes\1_{E}\right]=0$ implies that the choice of control $s\equiv0$ leaves the cost $J=\langle C\otimes\1,\rho\rangle$ unchanged since in this case $H_\mathrm{tot}$ commutes with $C\otimes\1_E$. 

\begin{proof}
Let us compute the PMP Hamiltonian at $t=t_f$. 
Note first that it follows from Eq.~\eqref{eq:pt_H_open} that 
\begin{align}
\langle p_f,\mathcal{K}_{C\otimes\1_E}\rho_f\rangle &=\langle\mathcal{K}_{C\otimes\1_E}^{\dag}p_f,\rho_f\rangle\notag\\
&=\langle\mathcal{K}_{C\otimes\1_E}C\otimes\1_E,\rho_f\rangle =0\ .
\label{eq:opencalc1}
\end{align}
The other two terms of the PMP control Hamiltonian Eq.~\eqref{eq:Ham_open} also vanish at $t=t_f$: using the same calculation as in Eq.~\eqref{eq:opencalc1} the second term vanishes because
of the assumption $\left[H_{I},C\otimes\1_E\right]=0$, and the third term
does as well because $[\1_S\otimes H_{E},C\otimes\1_E]=0$.
So we obtain $\mathbb{H}=s(t_f)\langle p_f,\mathcal{K}_{B\otimes\1_E}\rho_f\rangle = \lambda >0$. This implies that $\langle p_f,\mathcal{K}_{B\otimes\1_E}\rho_f\rangle >0$ and by continuity there must exist an interval $(t_f-\epsilon,t_f]$ for some $\epsilon>0$ such that 
$\langle p(t),\mathcal{K}_{(B-C)\otimes\1_E}\rho(t)\rangle>0$ for $t\in(t_f-\epsilon,t_f]$. Finally we must have $s^\ast\equiv1$ in this interval by Corollary \ref{control_general_open}. 
\end{proof}

Regarding the arc at the beginning we have instead:
\begin{thm}
\label{thm:open_beginning} 
Assume $s\in[0,1]$ is the optimal control
in an interval $[0,t_f]$ minimizing the cost Eq.~\eqref{eq:cost_open}
for Eq.~\eqref{eq:open}. Assume that $\left[\rho_{E},{H}_{E}\right]=0$
and that $\left[H_{I},\rho_{0}\otimes\rho_{E}\right]=0$. Assume furthermore
that $\left[B,\rho_{0}\right]=0$ and that the final time constraint
is active $\lambda>0$. Then the control satisfies $s^\ast(t)=0$ for $t\in[0,\epsilon)$
for some $\epsilon>0$. 
\end{thm}

Note that the assumption $\left[H_{I},\rho_{0}\otimes\rho_{E}\right]=0$
implies that the interaction alone does not modify the initial state
of the system.
\begin{proof}
We abbreviate the proof since it is very similar to the ones we presented above in more detail. 
Using the assumption $\left[H_{I},\rho_{0}\otimes\rho_{E}\right]=0$ and $\left[B,\rho_{0}\right]=0$, evaluating $\mathbb{H}(t=0)$
we obtain $\mathbb{H}(0)=\left(1-s(0)\right)\langle p(0),\mathcal{K}_{C\otimes\1_E}\rho(0)\rangle=\lambda>0$. This implies that
$\langle p(0),\mathcal{K}_{C\otimes\1_E}\rho(0)\rangle>0$ or equivalently that  $\langle p(0),\mathcal{K}_{-C\otimes\1_E}\rho(0)\rangle<0$. By continuity this in turn implies that there exist an interval $[0,\epsilon)$ for some $\epsilon>0$ such that $\langle p(t),\mathcal{K}_{(B-C)\otimes\1_E}\rho(t)\rangle
<0$ for $t\in [0,\epsilon)$. Finally we must have $s^\ast\equiv0$ in this interval by Corollary \ref{control_general_open}.
\end{proof}

We comment on the implications of the additional assumptions used in these theorems in Sec.~\ref{sec:conc}.

\subsection{Optimal control for quantum master equation dynamics}
\label{QME}

The treatment of the open system case in the previous subsection did not involve any approximations. On the other hand, we tacitly
assumed that the environment is \emph{finite dimensional}. This was helpful since all the results on optimal control which we have elaborated
upon in Sec.~\ref{OCP} and used so far, are classically stated
and proved for finite dimensional systems. Extending such results,
in particular concerning the PMP and the topology and continuity
of the reachable sets for infinite dimensional systems, is possible
and is a current area of research in control theory (see, e.g.,
Refs.~\cite{fiacca_existence_1998,kipka_pontryagin_2015}), although the
results in this area become considerably more technical. An alternative
we discuss in this subsection is to replace the Liouville-von Neumann
equation~\eqref{eq:open} with an approximate quantum master equation. This can be viewed as an investigation of the result of Sec.~\ref{OCLVN} when the environment dimension is sent to infinity. 

Without loss of generality we write $H_{I}=\sum_{\alpha}S_{\alpha}\otimes E_{\alpha}$, where $S_{\alpha}=S_{\alpha}^\dag$ and $E_{\alpha}=E_{\alpha}^\dag$ $\forall\alpha$. The goal is now to find a master equation for the dynamics of the system density matrix
in the case of a time-dependent system Hamiltonian. After the Born approximation and 
tracing out the environment, one arrives at a time dependent Redfield master
equation (see, e.g., the Schr\"odinger picture Redfield master equation (SPRME) derived in Ref.~\cite{campos_venuti_error_2018}). From this point there are multiple ways to proceed, e.g., by introducing an additional adiabatic approximation or an additional Markovian approximation, or both. These different paths, and exactly how they are taken, lead to a plethora of different master equations~\cite{childs_robustness_2001,PhysRevA.73.052311,albash_quantum_2012,campos_venuti_error_2018,Mozgunov:2019aa,Yamaguchi:2017vu,Dann:2018aa,nathan2020universal,Davidovic2020completelypositive,winczewski2021bypassing}. We next focus on two representative cases of master equations derived from first principles.

\subsubsection{Adiabatic Redfield Master Equation}

The Adiabatic Redfield Master Equation (ARME) is derived in Ref.~\cite{campos_venuti_error_2018}. It results from assuming that $t_f\gg\tau_B$, where $\tau_{B}$ is the environment time scale, and the adiabatic approximation $\mc{T}\exp\left[-i\int_{t-r}^t H_S(t')dt'\right] \approx e^{-i r H_S(t)}$, dropping a correction of $O\left((r/t_f)^2\right)$.
The ARME has the form of Eq.~\eqref{eq:general_ME} with a time-dependent Redfield
generator $\mathcal{L}$ given by
\bes
\label{eq:Redfield}
\begin{align}
\mathcal{L} & =\mathcal{K}_{H_S}+\mathcal{D}\label{eq:Redfield-a}\\
\mathcal{D}\rho & =\sum_{\alpha\beta}\int_{0}^{t_{\max}}dr\ G_{\alpha\beta}(r)\,\left[S_{\beta}(-r)\rho,S_{\alpha}\right]+\hc, 
\label{eq:Redfield-b} 
\end{align}
\ees
where the system Hamiltonian is as in Eq.~\eqref{eq:H_S}.
$G_{\alpha\beta}(t)$ is the environment correlation function
\begin{equation}
G_{\alpha\beta}(t)=\langle E_{\alpha}(t)E_{\beta}(0)\rangle = \Tr[E_{\alpha}(t)E_{\beta}(0)\rho_E]\ ,
\label{eq:environment_corr}
\end{equation}
 where $\left\langle X\right\rangle $ denotes the environmental thermal average
of $X$. When $G_{\alpha\beta}(r)$ decays exponentially, the relative error of the resulting dynamics due to the adiabatic approximation above is $O\left((\tau_{B}/t_f)^{2}\right)$.
Finally, 
\begin{equation}
S_{\beta}(-r)=e^{-irH_S(s)}S_{\beta}e^{irH_{S}(s)}\ .
\label{eq:A_r}
\end{equation}
The parameter $t_{\max}$ can either be set to $t_f$ or
infinity on account of the fact that the environment correlation function
decays very rapidly. The ARME is not in Gorini-Kossakowski-Sudarshan-Lindblad (GKSL) form~\cite{Gorini:1976uq,Lindblad:76,Chruscinski:2017ab}, hence does not generate a completely positive map. However, it generates non-Markovian dynamics, hence has a wider range of applicability than Markovian master equations, within its range of applicability~\cite{Mozgunov:2019aa}.

Crucially, the generator in Eq.~\eqref{eq:Redfield} with Eqs.~\eqref{eq:H_S}, \eqref{eq:environment_corr} and~\eqref{eq:A_r}
\emph{depends on time only through the control function $s(t)$}. This implies
that the PMP control Hamiltonian $\mathbb{H}=\langle p,\mathcal{L}\rho\rangle$
is constant and hence the PMP in the form of Theorem~\ref{adapt0} is directly applicable.\footnote{It can be shown that the conditions of Filippov's theorem (see Ref.~\cite[Th.~2.1]{fleming_deterministic_1975}) for the existence of the
optimal control solution are satisfied also in this case.} However, the control now enters nonlinearly in $\mathcal{D}$,
in particular in an exponential through Eq.~\eqref{eq:A_r}. As a
consequence it is not possible to derive the form of the control
on the nonsingular arcs, or even to determine simple equations for the
appearance of singular arcs. One can ask, however, what remains of
the results of the previous subsection. We do not have an analog of Theorem~\ref{thm:open_beginning}  for the initial arc. However, if we again make the assumption of Theorem~\ref{thm:open_end} that $\left[S_{\alpha},C\right]=0$ $\forall\alpha$, then the analog of this theorem for the final arc holds, even when the environment is infinite-dimensional, and under the approximations used to derive Eq.~\eqref{eq:Redfield}. However, instead of an arc, we obtain a bang only at a point:

\begin{thm}
\label{thm:Redfield}
Assume $s^*\in[0,1]$ is the optimal control in an interval $[0,t_f]$
minimizing the cost Eq.~\eqref{nuovocosto} for Eq.~\eqref{eq:general_ME}
with $\mathcal{L}$ given by Eq.~\eqref{eq:Redfield}. Assume further
that $\left[S_{\alpha},C\right]=0$ $\forall\alpha$ and that the
final time constraint is active ($\lambda>0$). Then the optimal control satisfies
$s^*(t_f)=1$. 
\end{thm}

\begin{proof}
It is convenient to write the Redfield dissipator as 
\bes
\begin{align}
\mathcal{D}\rho & =\sum_{\alpha\beta}\left(\big[W_{\alpha\beta}\rho,S_{\alpha}\big]+\big[S_{\alpha},\rho W_{\alpha\beta}^{\dagger}\big]\right)\label{eq:Redfield_dissi}\\
W_{\alpha\beta}(t) & =\int_{0}^{t_{\max}}dr\ G_{\alpha\beta}(r)\,S_{\beta}(-r)\ .
\end{align}
\ees
Using Eq.~\eqref{eq:Redfield_dissi} one obtains, for the adjoint
of $\mathcal{D}$:
\beq
\mathcal{D}^{\dag}X=\sum_{\alpha\beta}\left(W_{\alpha\beta}^{\dagger}\left[X,S_{\alpha}\right]+\left[S_{\alpha},X\right]W_{\alpha\beta}\right)
\label{eq:Ddag-Red}
\eeq
(see Appendix~\ref{app:calcs}).
 From the above expression and the assumption $\left[S_{\alpha},C\right]=0,\,\forall\alpha$
we obtain $\mathcal{D}^{\dag}C=0$. Using $p_f=-C$ we have $\langle p_f,\mathcal{D} \rho_f\rangle=-\langle\mathcal{D}^{\dag}C,\rho_f\rangle=0$.
Evaluating the PMP control Hamiltonian at the final time we obtain $\mathbb{H}(t_f)=s(t_f)\langle p_f,\mathcal{K}_{B}\rho_f\rangle = \lambda>0$. This implies that $\langle p_f,\mathcal{K}_{B}\rho_f\rangle>0$, and so $s(t_f)=1$ from the maximum principle. 
\end{proof}

Since, as argued in Sec.~\ref{shortening}, the size of the bang intervals is generically expected to shrink when the total system size increases, this result can be seen as a generalization of Theorem \ref{thm:open_end} when the environment dimension is sent to infinity and the bang interval at the end shrinks to a point. 

The implications of the additional assumptions used in Theorem~\ref{thm:Redfield} are discussed in Sec.~\ref{sec:conc}.

\subsubsection{Markovian, completely positive master equations}
The most significant drawback of the ARME is the violation of complete positivity, which means that the density matrix can develop unphysical, negative eigenvalues. Hence we also consider Markovian, completely positive master equations. There are a variety of such master equations derived from first principles under different assumptions. However, in most cases the generator $\mc{L}$ is explicitly time-dependent (e.g., the coarse-grained master equation (CGME)~\cite[Eq.~(22)]{Mozgunov:2019aa}, the master equation of Ref.~\cite[Eq.~(21)]{Yamaguchi:2017vu}, the non-adiabatic master equation (NAME)~\cite[Eq.~(16)]{Dann:2018aa}, and the universal Lindblad equation (ULE)~\cite[Eq.~(27)]{nathan2020universal}) and hence we cannot apply Theorem~\ref{adapt0}. 

In this subsection we give an example of a Markovian master equation derived from first principles where, like in the ARME case, the generator $\mc{L}$ depends on time only through the schedule $s(t)$. In this case the PMP can be applied in the simplified form described in Theorem~\ref{adapt0}.

Consider the ``geometric-arithmetic master equation'' (GAME)~\cite[Eq.~(46)]{Davidovic2020completelypositive}, which is claimed there to have a higher degree of accuracy than all the previous Markovian master equations. In the adiabatic limit it has the Schr\"odinger picture form
\bes
\begin{align}
\mc{L} &= \mc{K}_{H_S} + \mc{D}\\
\mc{D}\r &= \sum_j \left( [L_\a(s)\r,L_\a^\dag(s)] + [L_\a(s),\r L_\a^\dag(s)] \right)\ ,
\label{eq:Drho-GAME}
\end{align}
\ees
where $L_\a(s) = S_\a \circ \sqrt{\gamma(s)}$, the circle denotes the Hadamard (element-wise) product, $\g$ is the spectral density matrix [Fourier transform of the environment correlation function~\eqref{eq:environment_corr}] whose elements $\g_{nm}:=\g(\o_{nm}(s))$ depend on the instantaneous Bohr frequencies $\o_{nm}(s)=E_n(s)-E_m(s)$, where $H_S(s)|n(s)\>=E_n(s)|n(s)\>$, and the dependence on time is only through the schedule $s(t)$.\footnote{In writing these expressions we have adapted the results of Ref.~\cite{Davidovic2020completelypositive} to the adiabatic limit by using the instantaneous eigenbasis of $H_S$, and also performed the ``$\int^t \to \int^\infty$ approximation'' (otherwise $\gamma$ would have a $t$ dependence, which would prevent us from being able to apply Theorem~\ref{adapt0}); see Ref.~\cite{Davidovic2020completelypositive} for complete details.} The adjoint dissipator is now:
\begin{align}
\mathcal{D}^{\dag}X&=\sum_{\a}\big((S_\a \circ \sqrt{\gamma})^\dag\left[X,S_\a \circ \sqrt{\gamma}\right]\notag\\
&\qquad +[(S_\a \circ \sqrt{\gamma})^\dag,X]S_\a \circ \sqrt{\gamma}\big)\ .
\label{eq:Ddag-GAME}
\end{align}
Unfortunately, since $[S_\a,C]=0 \not\Rightarrow [S_\a \circ \sqrt{\gamma},C]=0$, this means that even if $[S_\a,C]=0$, we do \emph{not} obtain $\mc{D}^\dag C=0$ as in the ARME case, and hence the proof of Theorem~\ref{thm:Redfield} does not carry through.\footnote{Of course $[S_\a,C]=0\ \Rightarrow [S_\a \circ \sqrt{\gamma},C]=0$ when $\g_{nm}=c$ $\forall n,m$, but this is a highly nongeneric scenario.}  

While these arguments are not a proof that in general Markovian dynamics 
do not admit $s(t_f)=1$ as an optimal control solution, we conjecture that in fact, they do not. It thus appears that the ``counterintuitive'' appearance of the driver Hamiltonian at the end of the control interval is not a feature of the optimal schedule in the Markovian limit of open quantum systems. We revisit this point in Sec.~\ref{sec:conc}.

\section{Switching operator and analysis of the optimal control}
\label{Switchings}

In order to study the qualitative behavior of the optimal control
law, in particular its switching properties and the existence and
nature of the singular arcs, it is convenient to introduce one more
operator, besides the state $\rho$ and the co-state $p$, which we
call the \emph{switching operator}. The switching operator determines
the behavior of the optimal control, i.e., the points where there
is a switch between $s=0$ and $s=1$, and where there is
a singular arc. For clarity we focus on the closed system case of
Sec.~\ref{OCP} but our definitions and treatment naturally extend
with a change of notation to the open system case of Sec.~\ref{OCLVN}.

\subsection{Switching equation}
\label{sec:switch-eq}

The switching operator $S$ is the Hermitian operator defined as 
\begin{equation}
S:=i\left[p,\rho\right]\ .
\label{sw}
\end{equation}
In the closed system case the Liouvillian has the form $\mathcal{L}=\mathcal{K}_{H}$,
with the system Hamiltonian of Eq.~\eqref{eq:Ham_sys2}. One has the
following property:
\begin{equation}
\mathcal{K}_{H}\left(\left[X,Y\right]\right)=\left[\mathcal{K}_{H}(X),Y\right]+\left[X,\mathcal{K}_{H}(Y)\right]\label{eq:autom_property}
\end{equation}
 valid for any operators $X,Y,H$. Now, differentiating Eq.~\eqref{sw},
we obtain
\bes
\label{eq:S_ODE}
\begin{align}
\dot{S} & =i\left[\dot{p},\rho\right]+i\left[p,\dot{\rho}\right]\\
 & =i\left[\mathcal{K}_{H}p,\rho\right]+i\left[p,\mathcal{K}_{H}\rho\right]\\
 & =i\mathcal{K}_{H}\left(\left[p,\rho\right]\right)=\mathcal{K}_{H}S\ ,
\end{align}
\ees
 where in the third equality we used Eq.~\eqref{eq:autom_property}.
Thus, $S$ satisfies the same equation as $\rho$ and $p$. To understand
why $S$ determines the optimal control switching times, let us define
\bes
 \label{eq:x_C}
 \begin{align}
x_{C} & :=\langle p,\mathcal{K}_{C}\rho\rangle =-i\Tr\left(p\left[C,\rho\right]\right) \\
 & =-i\Tr\left(C\left[\rho,p\right]\right) =\Tr\left(CS\right)=\langle C,S\rangle\ ,
 \label{eq:x_C-b}
 \end{align}
 \ees
 and similarly 
\begin{equation}
x_{B}:=\langle p,\mathcal{K}_{B}\rho\rangle=\langle B,S\rangle\ ,
\label{eq:x_B}
\end{equation}
so that the PMP control Hamiltonian Eq.~\eqref{HamilP-1} can be written in a form closely resembling the Hamiltonian~\eqref{eq:Ham_sys2-b}:
\begin{equation}
\mathbb{H}=x_{C}+s\left(x_{B}-x_{C}\right)\ .
\label{eq:new_controlH}
\end{equation}
The quantity $x_{B}-x_{C}=\langle B-C,S\rangle$, that
is, the orthogonal component of $S$ along $B-C$, regulates the switches
of the candidate optimal control. By the same PMP argument we have used repeatedly in our proofs, when $x_{B}-x_{C}<0$ we have $s\equiv0$, and 
when $x_{B}-x_{C}>0$ we have $s\equiv1$. The switch occurs when $x_{B}-x_{C}=0$,
while a singular arc occurs when $x_{B}-x_{C}\equiv0$ for an interval
of positive measure. 

The initial condition $S_{0}$ of the switching
operator determines the optimal control candidate $s$ uniquely. This initial condition is not completely arbitrary. In particular,
under the assumptions of Theorem~\ref{Bradygen} the following holds.
Since $s(0)=0$, we have $x_{C}(0)=\langle C,S_{0}\rangle=\lambda$.
Furthermore, since $[B,\rho_{0}]=0$, we have 
\begin{align}
x_{B}(0)&=\langle B,S_{0}\rangle=i\Tr(B[p_0,\rho_{0}])\notag\\
&=i\Tr\left(p_0[\rho_{0},B]\right)=0\ .
\label{eq:xB0}
\end{align}
At the final time $t_f$, since $s(t_f)=1$, we have from Eq.~\eqref{eq:new_controlH}
$x_{B}(t_f)=\lambda$, while using Eq.~\eqref{3plusB} 
we have 
\begin{align}
x_{C}(t_f)&= \<C,S_f\> = i\Tr\left(C[p_f,\r_f]\right) \notag \\
&= i\Tr\left(\r_f[C,p_f]\right)=0\ .
\label{eq:xCf}
\end{align}
Thus, at $t=0$ we have $x_B(0)=0$ and in the initial arc $s\equiv0$ and $x_C\equiv \lambda$ by Eq.~\eqref{eq:new_controlH}. The next arc can be either nonsingular with $s\equiv1$ ($x_B>x_C$) or a singular arc with $x_B \equiv x_C$. Either way, at the switching point we must have $x_B=\lambda$, hence we reach the point $(x_{C},x_{B})=(\lambda,\lambda)$ at the end of the first arc. 
When $(x_{C},x_{B})=(\lambda,\lambda)$
there is either a switch to an arc with $s\equiv1$, or a return to
$s\equiv0$, or a singular arc where $(x_{C},x_{B})\equiv(\lambda,\lambda)$.
On this arc $s$ is unspecified, but nonetheless certain equations need to be satisfied and they can be used to obtain information on the dynamics on such singular arcs (see Appendix \ref{app:SA}).
Note that every switching event, whether from a bang arc to a singular arc or \textit{v.v.}., or from a bang arc to another bang arc, happens at $(x_{C},x_{B})=(\lambda,\lambda)$.
When $s\equiv1$, from Eq.~\eqref{eq:new_controlH} 
we have that $x_{B}$ is constant, while $x_{C}$ is allowed to change.
Therefore the optimal control can be described schematically as in
Fig.~\ref{switchfig} in the $(x_{C},x_{B})$ plane. 

\begin{figure}[h]
\hspace{-1cm}
\includegraphics[width=\columnwidth]{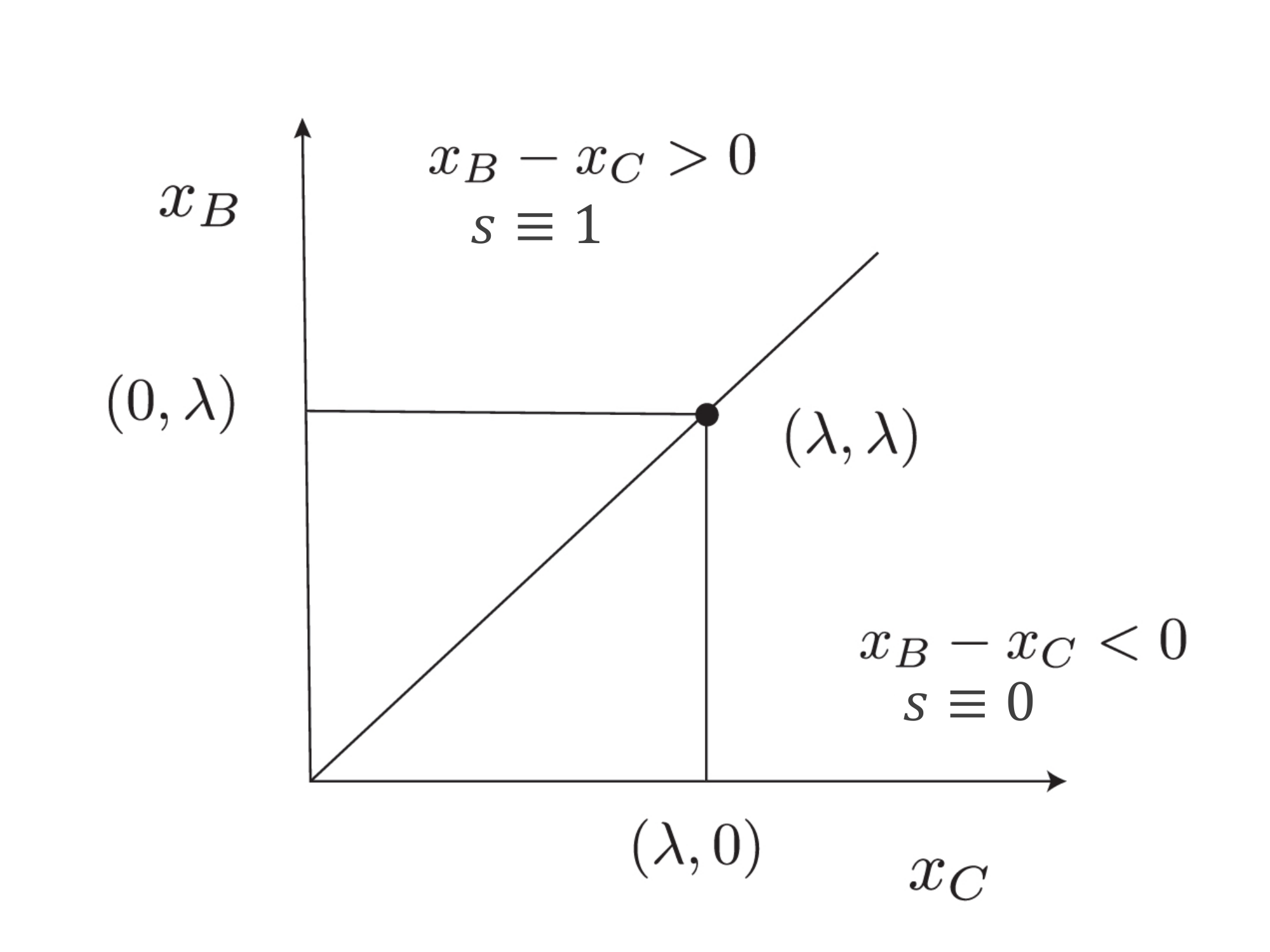} 
\caption{Switching diagram for optimal control candidates. Starting from
the point $(\lambda,0)$ the dynamics stay on the line $x_C=\lambda$. In principle even negative values of $x_B$ can be attained, but eventually the point 
$(\lambda,\lambda)$ must be reached, followed by an alternation of vertical and
horizontal lines (where in principle one can also extend to negative values of $x_C$) through $(\lambda,\lambda)$ with time
intervals where the dynamics do not move from $(\lambda,\lambda)$,
representing the singular arcs. In the last interval, the dynamics follow a
horizontal line (with $s\equiv1$) reaching the point $(x_{C},x_{B})=(0,\lambda)$.}
\label{switchfig}
\end{figure}

In the case where the final time is not active, i.e., when we do not have the guarantee $\lambda>0$ in Eq.~\eqref{eq:new_controlH},
the above reasoning can still be applied to conclude $x_{B}(0)=0$
and $x_{C}(t_f)=0$. Furthermore, if $\lambda=0$, from Eq.~\eqref{eq:new_controlH},
we obtain $(1-s(0))x_{C}(0)=0$ and $s(t_f)x_{B}(t_f)=0$, in addition
to the maximization condition Eq.~\eqref{MaxCon}.

\subsection{Shortening of  nonsingular arcs}
\label{shortening}

\subsubsection{Dependence on $n$}

Here we give a heuristic argument that explains why the terminal arcs ($s\equiv0$ and $s\equiv1$) for the optimal
control become shorter as the number $n$ of spins (or qubits) increases. In fact the heuristic
holds also for intermediate arcs taking place anywhere along the optimal
trajectory.

Consider a bang arc where $s(t)= 0$ for $t\in\left[t_{0},t_{1}\right]$.
Equation~\eqref{eq:S_ODE} for the switching operator in this
region is $\dot{S}=\mathcal{K}_{C}S$ with the initial condition $S(t_{0})=S_{0}$
for some $S_{0}$. The solution in this interval is $S(t)=e^{-i(t-t_{0})C}S_{0}e^{i(t-t_{0})C}.$
The coordinate $x_{C}$ equals $\lambda$ in the interval: $x_{C}=\lambda=\langle C,S(t)\rangle=\langle C,S_{0}\rangle$.
The switching happens when $x_{B}=\lambda$, i.e., at the first solution
$t_{1}$ of $\langle B,S(t_{1})\rangle=\langle C,S_{0}\rangle$. 
More explicitly, the interval of the bang arc $\Delta t:=t_{1}-t_{0}$
is given by the first solution of 
\begin{equation}
x_{B}(\Delta t):=\Tr\left(Be^{-i\Delta tC}S_{0}e^{i\Delta tC}\right)=\Tr\left(CS_{0}\right)\ .\label{eq:switching_time}
\end{equation}
Using the spectral resolution $C=\sum_{k}E_{k}\Pi_{k}$ (with eigenvalues $E_k$ and eigenprojectors $\Pi_k$), the left-hand hand side of Eq.~\eqref{eq:switching_time} can be written as
\begin{equation}
x_{B}(\Delta t)=\sum_{kl}M_{kl}e^{-i\Delta t\omega_{kl}}\ ,
\label{eq:trigon_pol}
\end{equation}
with amplitudes $M_{kl}:=\Tr\left(B\Pi_{k}S_{0}\Pi_{l}\right)$ and
Bohr frequencies $\omega_{kl}=E_{k}-E_{l}$. The function $x_{B}( \Delta t)$
is a real trigonometric polynomial with $O(d^{2})$ terms
(where $d$ is the Hilbert space dimension) starting from $\Tr(B S_0)$ at $\Delta t=0$,  
and reaching $\lambda>0$ at $\Delta t$.
Now, as the number of qubits $n$ increases, both $d=2^{n}$ and the
frequencies increase. Both of these facts contribute to making $x_{B}(t)$
oscillate faster. As a consequence, the solution $\Delta t$ of Eq.~\eqref{eq:switching_time}
tends to decrease with $n$. The same considerations hold for the
case of an $s\equiv1$ arc with $B$ and $C$ interchanged. See Appendix~\ref{app:shortening} for additional comments.
Note that this heuristic applies both to the initial and final bang arcs, as well as to possible intermediate bang arcs if they are present. 

\subsubsection{Dependence on $t_f$}

It was concluded in Ref.~\cite[Sec.~S3]{brady_optimal_2021} that ``these bangs should become smaller and smaller as $t_f$ is increased. Eventually in the true $t_f\to \infty$ adiabatic limit, these bangs disappear recovering the standard form
expected for quantum adiabatic computing.'' However, there is in fact no guarantee that the optimal control coincides with the adiabatic path even in this limit, since the adiabatic theorem provides a sufficient, but not a necessary condition for convergence to the minimum of the cost function. Indeed, it is easy to construct a counterexample, as we now do. First note that, as we have seen, when $t_f<t_c$, the optimal schedule always starts with a bang  $s\equiv 0$ and ends with a bang $s\equiv 1$ (provided the initial state commutes with $B$). We have not addressed the question of uniqueness of this optimal schedule, which we leave for future work. However, when $t_f>t_c$, the optimal schedule is certainly not unique, as one has the possibility of ``wasting time'' by adding a bang $s\equiv 0$ at the end (thus applying $C$ there), or by adding a bang $s\equiv 1$ at the beginning (thus applying $B$ there), or both. The resulting schedules do not resemble the smooth adiabatic schedule interpolating slowly from $s(0)=1$ to $s(t_f)=0$.

\section{Summary and Discussion}
\label{sec:conc}

The quest to discover the optimal schedule for quantum optimization algorithms such as QA and QAOA naturally leads to the use of optimal control theory via Pontryagin's principle. Previous work concluded that QAOA is optimal~\cite{Yang:2017aa}, but a more careful analysis showed that in fact a hybrid bang-anneal-bang protocol is generally optimal for closed systems when not enough time is allowed for the
desired state to be reached perfectly~\cite{brady_optimal_2021}. Here we confirmed this result using a density matrix approach, which both generalizes the analysis to mixed states and simplifies it since it makes the cost-function linear in the state. We also showed that the assumption that $t_f$ is smaller than the critical time $t_c$ needed to reach the ground state of $C$ exactly, is necessary but not sufficient for the result of Ref.~\cite{brady_optimal_2021}, by giving a counterexample to the latter. 

We introduced a switching operator and found its equation of motion, which characterizes the points at which the optimal schedule switches between the two different types of bang arcs and the anneal arc. 
In Theorem~\ref{conclusione} we gave the explicit optimal schedule for the example of a single spin-$1/2$ particle, which consists of two bangs of equal duration.

Using the density matrix formulation we extended the theory to the open system setting, both for the exact reduced system dynamics in the case of a finite-dimensional environment, and under the approximation of dynamics governed by a master equation due to coupling to an infinite-dimensional environment. We proved that in the first setting, depending on additional assumptions concerning the initial states of the system and the environment and their interaction, either an anneal-bang (Theorem~\ref{thm:open_end}) or bang-anneal (Theorem~\ref{thm:open_beginning}) schedule is optimal. 

In the second setting (infinite-dimensional environment) we considered both an adiabatic Redfield equation accounting for non-Markovian dynamics but without a complete positivity guarantee, and a completely positive Markovian master equation. In the former (Redfield) case 
we could only prove that the optimal schedule terminates with the driver Hamiltonian, i.e., $s(t_f)=1$ (Theorem~\ref{thm:Redfield}). 
One could interpret this result  as a manifestation of the phenomenon of the shortening of the nonsingular arcs as the \emph{total} system (i.e., the subsystem plus its environment) size increases. Indeed, in the Redfield case the environment Hilbert space dimension is infinite, which is consistent with the bang arc having shrunk down to a point. In the fully Markovian case, even this last remnant of the bang-arc was not recovered, as we found no evidence of natural conditions under which $s(t_f)=1$ holds.

Let us now comment on the differences between these theorems and their closed system counterpart, Part (ii) of Theorem~\ref{Bradygen}.
Regarding Theorem~\ref{thm:open_end} concerning the final arc, the main change is the addition of the assumption that the interaction Hamiltonian commutes with the cost function, i.e., $\left[H_{I},C\otimes\1_{E}\right]=0$. Writing the interaction in the general form $H_I = \sum_\alpha S_\alpha\otimes E_\alpha$, where $S_\alpha$ and $E_\alpha$ are system and environment operators, respectively, the assumption is equivalent to $[S_\alpha,C]=0$ $\forall \alpha$, which is the same assumption as in Theorem~\ref{thm:Redfield}. Thus $C$ must belong to the commutant of the algebra generated by the set $\{S_\alpha\}$, i.e., $C\in\text{Alg}\{S_\alpha\}'$.\footnote{The commutant of an algebra $\mathcal{A}=\text{Alg}\{A\}$ is defined as the set $\mathcal{A}':=\{X | [X,A]=0 \ \forall A\in \mathcal{A}\}$.} For example, if $S_\alpha = \sum_{i=1}^n \sigma_{i}^{\alpha}$ for a system of $n$ qubits, where $\alpha\in\{x,y,z\}$, i.e., the collective decoherence case~\cite{Zanardi:98a,Lidar:1998fk}, then, if $C$ is at most a two-body interaction, it follows that it must be of the Heisenberg interaction form: $C = \sum_{ij} J_{ij} \boldsymbol{\sigma}_i\cdot\boldsymbol{\sigma}_j$, where $J_{ij}$ are constants~\cite{Bacon:2000qf}.
Or, for a classical target Hamiltonian arising in optimization such as the Ising-type Hamiltonian mentioned in Sec.~\ref{sec:background}, this means that the interaction must be of the pure-dephasing type, i.e., $S_\alpha \propto \sigma^z$ (or products of $\sigma^z$ over different qubits). This is a realistic model, e.g., for superconducting qubits undergoing flux noise~\cite{Kjaergaard_2020}.

Regarding Theorem~\ref{thm:open_beginning} concerning the initial arc, the main change is the addition of the two assumptions that (i) the environment Hamiltonian commutes with the environment's initial state ($\left[{H}_{E},\rho_{E}\right]=0$),
and (ii) the interaction Hamiltonian commutes with the joint system-environment initial state ($\left[H_{I},\rho_{0}\otimes\rho_{E}\right]=0$). The first of these is natural and is known as the stationary environment assumption~\cite{Breuer:book}. It is satisfied, e.g., if the environment is in thermal equilibrium, i.e., in the Gibbs state: $\rho_E \propto e^{-\beta H_E}$, where $\beta$ is the inverse temperature. The second assumption means that $\rho_{0}\in\text{Alg}\{S_\alpha\}'$ and $\rho_E\in\text{Alg}\{E_\alpha\}'$. This assumption is the least natural of the ones we have encountered so far. E.g., it is clearly violated in the standard quantum annealing setting where $\rho_0$ is the ground state of a transverse field $-\sum_i\sigma_i^x$ and $S_\alpha \propto \sigma^z$. Even the condition $\rho_E\in\text{Alg}\{E_\alpha\}'$ is not very natural. For example, for a bosonic environment one typically has $E_\alpha$ as the position operator of an oscillator, while $H_E$ might be the number operator, in which case the Gibbs state $\rho_E$ would not commute with $E_\alpha$. 

We note that while the conditions given in Theorem~\ref{thm:open_beginning} are sufficient, we do not know if they are  necessary, which we thus leave as an open problem. 
We conjecture that the initial bang does not appear as a feature of optimal schedules for open systems coupled to an infinite-dimensional environment.
This state of affairs would be reminiscent of the existence of an arrow of time for open systems, which breaks the symmetry between the initial and final times (see, e.g., Ref.~\cite{campos_venuti_error_2018} for a similar effect in the pure QA setting).  
On the other hand, given the naturalness of the sufficient conditions under which a final bang arc (Theorem~\ref{thm:open_end}) or a schedule terminating with the driver Hamiltonian (Theorem~\ref{thm:Redfield}) are optimal, such schedules may find utility in the design of quantum algorithms for optimization problems in the setting of open quantum systems. This is true in particular for systems that are well described by the adiabatic Redfield master equation, e.g., superconducting flux qubits used for quantum annealing~\cite{Harris:2010kx,Yan:2016aa,smirnovTheory18,khezri2020annealpath}. 

However, the conditions under which the adiabatic Redfield master equations hold need not apply in general, e.g., for Hamiltonians that arise naturally in systems such as Rydberg atoms or transmons, which have been used to demonstrate QAOA~\cite{zhou2018quantum,Harrigan:2021we}. Especially for quantum optical systems such as Rydberg atoms, the Markovian limit may be more appropriate, and we have not found evidence of the optimality of an initial or final bang arc in this limit, or even the optimality of $s(t_f)=1$. 

While our analysis does not strictly rule out bang-type schedules for open systems coupled to infinite-dimensional environments, we conjecture that they are indeed not a feature of optimal schedules in this case, primarily due to the shortening of arcs in the open system setting. If this could be confirmed, it would mean that after all, continuous annealing-type schedules are optimal for optimization purposes when using open quantum systems, which would have implications for all NISQ-era optimization algorithms. 

Finally, we remark that throughout this work we have used a simplified form of the PMP as described in Theorem~\ref{adapt0}. A more general PMP for time-dependent dynamics is described in Ref.~\cite{fleming_deterministic_1975}.  The main difference is that $\mathbb{H}\equiv \lambda$ [Eq.~\eqref{forlanda}] is no longer valid in the given form. This equation is, in fact, a special case of another one which contains the derivative of the dynamics with respect to $t$, and this term vanishes when the dynamics are not explicitly time-dependent, as in our case, where we considered adiabatic time-dependent master equations. The time-dependence of the dynamics in typical quantum master equations~\cite{childs_robustness_2001,PhysRevA.73.052311,albash_quantum_2012,campos_venuti_error_2018,Mozgunov:2019aa,Yamaguchi:2017vu,Dann:2018aa,nathan2020universal,Davidovic2020completelypositive,winczewski2021bypassing} is, however, different from the one of models usually encountered in classical control theory, and therefore further study is required before the PMP can be applied to a broader class of quantum master equations.

\begin{acknowledgments}
LCV's and DL's research is based upon work (partially) supported by the Office of
the Director of National Intelligence (ODNI), Intelligence Advanced
Research Projects Activity (IARPA) and the Defense Advanced Research Projects Agency (DARPA), via the U.S. Army Research Office
contract W911NF-17-C-0050. DD's research was supported by the NSF under Grant ECCS
1710558. DL's research was also sponsored by the Army Research Office and was accomplished under Grant Number W911NF-20-1-0075. The views and conclusions contained herein are those of the authors and should not be interpreted as necessarily
representing the official policies or endorsements, either expressed or
implied, of the ODNI, IARPA, DARPA, or the U.S. Government. The U.S. Government
is authorized to reproduce and distribute reprints for Governmental
purposes notwithstanding any copyright annotation thereon.
\end{acknowledgments}

\appendix

\section{Examples of reachability/unreachability of the ground state of $C$}
\label{app:reachability}

Here we provide two examples, one where Prop.~\ref{prop:0} guarantees reachability of the ground state of $C$ despite the algebra generated by $B$ and $C$ being smaller than the full $su(d)$, another illustrating that $[B,C]\ne 0$ is not a sufficient condition. 

\subsection{First example}
Let $B= \sigma^z \otimes \1$, $C = \sigma^x \otimes  \sigma^x+ \sigma^y \otimes  \sigma^y+ \sigma^z \otimes  \sigma^z$. Both $B$ and $C$ commute with $M=\sigma^z\otimes \1+\1 \otimes \sigma^z$. We know that the ground state of $C$ is a singlet $\frac{1}{\sqrt{2}}(\ket{\downarrow\uparrow}-\ket{\uparrow\downarrow})$ with total spin zero and hence belongs to the sector $M=0$. This statement follows from a theorem of Marshall~\cite{marshall_antiferromagnetism_1955,lieb_ordering_1962} and holds for  general antiferromagnetic Heisenberg models defined on a bipartite lattice. In particular, it does not require knowledge of the ground state. 
The projector $P_0$ is given by 

\begin{equation}
P_{0}=\ketb{\downarrow\uparrow}{\downarrow\uparrow}+\ketb{\uparrow\downarrow}{\uparrow\downarrow} \ . 
\end{equation}

In $\mathrm{Ran}(P_0)$ one has:
\bes
\begin{align}
P_{0}BP_{0} & =\sigma^{z}\\
P_{0}CP_{0} & =-\1+2\sigma^{x}\ .
\end{align}
\ees
Therefore in $\mathrm{Ran}(P_0)$, $B$ and $C$ generate the full $su(2)$ algebra and any state in $\mathrm{Ran}(P_0)$ can be reached starting from any state in $\mathrm{Ran}(P_0)$, in particular the ground state of $C$. 

\subsection{Second example}
It is straightforward to find examples where the ground state of $C$ cannot be reached even when $[B,C]\neq 0$. 
Consider, e.g., $B= \sigma^x \otimes \1$ and $C= \sigma^z \otimes \1 + \1 \otimes \sigma^z$. In this case the Lie algebra generated by $B$ and $C$ is $su(2)\otimes u(1) \neq su(4)$ and only the first qubit can be fully steered anywhere on the Bloch sphere. As a consequence, the ground state $\ket{\downarrow\downarrow}$ of $C$ cannot be reached unless one starts with a state of the form $\rho_0 = \tilde{\rho}\otimes \ketb{\downarrow}{\downarrow}$, with $\tilde{\rho}$ being an arbitrary single qubit state.

\section{General results on optimal control; the Pontryagin Maximum Principle}
\label{app:A}

We review here standard results in optimal control theory emphasizing
a geometric viewpoint and the results needed for the applications
in the main body of the paper.

\subsection{Setup}

In optimal control theory (see, e.g., Ref.~\cite{fleming_deterministic_1975})
one considers a general control system 
\begin{equation}
\dot{x}=f\left(x,s\right),\qquad x(0)=x_{0},
\label{1}
\end{equation}
with $x\in\mathbb{R}^{N}$,\footnote{In our case $x$ is either a wavefunction $\psi$ ($N=n$) or a density matrix $\rho$ ($N=n^2$), after the coordinatization described in Appendix~\ref{app:real}.} the control $s$ with values from a compact
subset $S\subseteq\mathbb{R}^{M}$, $f$ a smooth map $\mathbb{R}^{N}\times\mathbb{R}^{M}\rightarrow\mathbb{R}^{N}$ that does not depend explicitly on time. The terminal cost (of \emph{Mayer type}\footnote{In control theory, problems with a cost that depends only on
the final values $t_f$ and $x(t_f)$, are called of Meyer type. In our case the cost does not depend explicitly on time so we simply consider cost of the form $\phi(x(t_f))$ to simplify the notation.}) 
\begin{equation}
J:=\phi\left(x(t_f)\right)\ ,
\label{cos}
\end{equation}
is to be minimized at the terminal (final) time $t_f$, where $\phi$ is a smooth function.
In particular, and this is the case that interests us, one can fix $t_f$ so that $J$ only depends on the final state $x(t_f)$. 

The geometric approach to the necessary condition of optimal control
(see, e.g., Ref.~\cite{agrachev_control_2004}) is based on the concept
of a {reachable set} (or {attainable set}) for Eq.~\eqref{1} with values of the control in $S$, which we denote by
$\mathfrak{R}_{t}$. The set $\mathfrak{R}_{t}$ is the set of values
for the state $x$ that can be reached (from $x_{0}$) at time \emph{exactly}
$t$ for control functions with values in the set $S$. With this
definition, the minimum of the cost $J$ in Eq.~\eqref{cos} is the minimum
of the function { $\phi$}  over $\mathfrak{R}_{t_f}$. One also defines
the reachable set $\mathfrak{R}_{\leq t_f}:=\cup_{0\leq t\leq t_f}\mathfrak{R}_{t}$,
and if $\mathfrak{R}_{t}$ is non decreasing with $t$, $\mathfrak{R}_{t_f}=\mathfrak{R}_{\leq t_f}$.
This is the case for Eq.~\eqref{eq:vonNeum} if one assumes,
as we do, that $\rho_{0}$ commutes with $B$ (in this case the system can remain in the state
$\rho_{0}$ for an arbitrary length of time with the choice $s\equiv1$). If the set of admissible controls $\tilde{S}$  is compact (as assumed in
the main text) and under general conditions on the map
$f$ in Eq.~\eqref{1}, which are also satisfied in our cases, Filippov's
theorem (see, e.g., \cite[Th.~10.1]{agrachev_control_2004})
states that the reachable sets are \emph{compact} and this implies
the existence of the minimum of the function $\phi$ and therefore
of the optimal control.\footnote{Filippov's existence theorem can also be applied in the version of
Theorem~2.1 in \cite[Ch.~III]{fleming_deterministic_1975}
if we assume [as for Eq.~\eqref{eq:vonNeum}] that $f$ in Eq.~\eqref{1}
is linear in both the control $s$ and state $x$, the cost $J$ in Eq.~\eqref{nuovocosto} is smooth, and the set $\tilde{S}$  is compact. Application
of this theorem also takes into account that $\rho$ in Eq.~\eqref{nuovocosto}
is bounded. The only condition that is not directly satisfied for Theorem
2.1 in \cite[Ch.~III]{fleming_deterministic_1975} is condition
(c) for which we can, however, use the alternative (c') of Corollary
2.2 in \cite[Ch.~III]{fleming_deterministic_1975}, i.e.,
the existence of a constant $\mu_{1}$ such that $\{x_{1}\,|\,J=\phi(x_{1})\leq\mu_{1}\}\subseteq\mathbb{R}^{n}$
is compact.} The introduction of the concept of a reachable set effectively reduces
the optimal control problem to a \emph{static} optimization problem
for the function $\phi$, where the set of possible dynamics is described
by the reachable set $\mathfrak{R}_{t}$, that roughly separates
the minimization problem from the analysis of the dynamics.

\subsection{The Pontryagin Maximum Principle (PMP)}
\label{subsec:PMP}

The basic necessary conditions of optimality are given by the Pontryagin
Maximum Principle (PMP), which we restate below in a more general formulation than in the main text, but in a context relevant for the problem of
interest to us, i.e., a fixed final time and a free final state. Assume
that $s^{*}=s^{*}(t)$ is the optimal control function and $x^{*}=x^{*}(t)$
the optimal trajectory. We shall refer to $(x^{*},s^{*})$ as an \emph{optimal
pair}. We have the following.
\begin{thm}
\label{case1} Assume that $(x^{*},s^{*})$ is an optimal pair. Then
there exists a nonzero vector of functions $p=p(t)\in\mathbb{R}^{n}$ called the {\em
co-state} that satisfies the terminal problem 
\begin{equation}
\dot{p}^{T}=-p^{T}\frac{\partial f}{\partial x}\left(x^{*}(t),s^{*}(t)\right),\quad p^{T}\left(t_f\right)=-\frac{\partial\phi}{\partial x}\left(x^{*}(t_f)\right),\ 
\label{3}
\end{equation}
with $f$ defined in Eq.~\eqref{1} and $\phi$ defined in Eq.~\eqref{cos}. Furthermore, define the Hamiltonian
function 
\begin{equation}
\mathbb{H}\left(p,x,s\right):=p^{T}f\left(x,s\right)\ .
\label{4}
\end{equation}
Then we have (maximum principle): 
\begin{equation}
\mathbb{H}\left(p,x^{*},s^{*}\right)=\max_{v\in \tilde{S}}\mathbb{H}\left(p,x^{*},v\right)\ ,
\label{5}
\end{equation}
(where $\tilde{S}$ is the set of admissible controls) and 
\begin{equation}
\mathbb{H}\left(p(t),x^{*}(t),s^{*}(t)\right)\equiv \lambda,
\label{6}
\end{equation}
 for a constant $\lambda$. 
\end{thm}

The constant $\lambda$ describes the dependence of the optimal cost
on the terminal time $t_f$. To see this, given the optimal control
$s^{*}$ defined in $[0,t_f]$, 
calculate the variation of the
cost with this control at $t=t_f$, 
\begin{align}
\left.\frac{d}{dt}\right|_{t=t_f}\phi(x(t)) & =\frac{\partial\phi}{\partial x}(x(t_f))f\left(x(t_f),s(t_f)\right)\nonumber \\
 & =-p^{T}\left(t_f\right)f\left(x(t_f),s(t_f)\right)\nonumber \\
 & =-\lambda\ ,
 \label{variacost}
\end{align}
where we used Eqs.~\eqref{1}, \eqref{3} and~\eqref{6}. In particular, $\lambda>0$
indicates that it is possible to lower the cost by increasing the
time or, in other words, the constraint $t=t_f$ is \emph{active}.
If $\lambda=0$ the constraint on the final time is not active.   If
$\lambda<0$, the above calculation shows that the cost is actually
\emph{increasing} with $t$ at $t=t_f$.

If there exists a value $s_0$ in the admissible control set $\tilde{S}$  such that
$f(x_{0},s_0)=0$ in Eq.~\eqref{1}, we can show that we must have $\lambda\ge 0$.
The argument is as follows. Assume $\lambda<0$ in Eq.~\eqref{variacost}.
Since $\phi(x(t))$ is increasing in $t$, there exists an $\epsilon>0$
such that $\phi(x(t_f-\epsilon))<\phi(x(t_f))=J$. Now construct
the following control function 
\beq
s_{1}(t)=\begin{cases}
s_0 & t\in[0,\epsilon]\\
s^{*}(t-\epsilon) & t\in(\epsilon,t_f]\ ,
\end{cases},
\eeq
i.e., $s_1$ leaves the cost unchanged in the first interval of time $[0,\epsilon]$ and then follows the optimal schedule shifted by $\epsilon$. Let us denote by $\phi_1(t)$ the cost function at time $t$ obtained with control $s_1$. 
Then, by construction, $\phi_1(t_f)=\phi(x(t_f-\epsilon))<\phi(x(t_f))$, 
which contradicts the fact that $s^{*}$ is optimal. Therefore
$\lambda$ cannot be negative. Summarizing, we have:

\begin{prop}
\label{landamagz} 
Assume there exists a value $s_0\in \tilde{S}$  such
that $f(x_{0},s_0)=0$ in Eq.~\eqref{1}. Then $\lambda\geq 0$ in Eq.~\eqref{6},
and  $\lambda>0$ if and only if the constraint on the final
time $t_f$ is active. 
\end{prop}

For the problem of interest in the main body of the paper the above
assumption on the existence of the value $s_0\in \tilde{S}$  is valid in
Theorems~\ref{Bradygen} and \ref{thm:open_beginning}, since
we assume that the initial condition commutes with the Hamiltonian
$B$ (i.e., we can choose $s_0=1$). Furthermore, we note, concerning reachable sets, that (i) because of the existence of such a value $s_0\in\tilde{S}$, we have ${\cal R}_{\leq t}={\cal R}_{t}$; (ii) from
standard result in control theory, (e.g., Ref.~\cite[Th.~1]{ayala_about_2017} and references therein) we know that for the bilinear
class of models [such as Eq.~\eqref{eq:vonNeum}], the sets ${\cal R}_{\leq t}={\cal R}_{t}$
are compact and continuous with $t$ with respect to the Hausdorff
metric. 

Since the optimal cost $J$ is the minimum of the continuous function $\phi=\phi(x)$ on $\mathfrak{R}_{t_f}$, it  depends continuously on
$t_f$, and since $\mathfrak{R}_{t_f}$ is nondecreasing with
$t_f$, the optimal cost is \emph{nonincreasing} with $t_f$.
That the constraint on the final time is active means that
the cost is actually strictly \emph{decreasing}. Notice in particular
that the function $\phi = \Tr(C\rho)$ giving the cost in Eq.~\eqref{nuovocosto} is
\emph{linear} in the state ($\rho$); hence the minimum is necessarily
achieved on the \emph{boundary} of the reachable set.\footnote{For a linear function $\phi(x)=k^{T}x$, if the minimum is achieved at an
interior point setting the derivatives equal to zero would imply $k=0$.} Therefore $\lambda>0$ in Eq.~\eqref{6} implies that at the optimal
final point, the boundary of the reachable set $\mathfrak{R}_{t_f}$
``moves'' in such a way so as to make the cost decrease.

\section{Coordinatization in terms of an orthonormal real matrix basis}
\label{app:real}

The application to quantum systems of the PMP, which is typically formulated over real vector spaces as in Appendix~\ref{app:A}, has to account for the fact that in the quantum case the equations are complex-valued. To show how this can be done we start with some basic preliminaries.

We denote the Hilbert-Schmidt scalar product between operators $A$ and $B$ acting on an $n$-dimensional Hilbert space $\mathcal{H}$ by  $\langle A,B\rangle :=\Tr(A^{\dagger}B)$.
For superoperators $\mathcal{L}$, we denote the Hilbert-Schmidt adjoint
of $\mathcal{L}$ by $\mathcal{L}^\dagger$, which is defined via 
\beq
\langle \mathcal{L}^\dag(A),B\rangle := \langle A,\mathcal{L}(B)\rangle\ \ \forall A,B\ . 
\eeq

We now choose an orthonormal basis $\{F_{j}\}$ for the \emph{real} vector space of Hermitian $n\times n$
matrices with Hilbert-Schmidt inner product $\langle A,B\rangle =\Tr(A B) $. Let us ``coordinatize'' $X=\sum_{j}\boldsymbol{X}_{j}F_{j}$ in this basis, where henceforth we use the notation $\boldsymbol{X}=\{\boldsymbol{X}_{j}\}$ for the vector of real-valued coordinates of the operator $X$. Then: 
\begin{align}
\langle A,B\rangle&=\sum_{jk}\boldsymbol{A}_{j}\boldsymbol{B}_{k}\Tr(F_{j}F_{k})=\sum_{jk}\boldsymbol{A}_{j}\boldsymbol{B}_{k}\delta_{jk}\notag \\
&=\boldsymbol{A}^{T}\boldsymbol{B}\ ,
\label{eq:<A,B>}
\end{align}
so that in these coordinates the inner product $\langle A,B\rangle$
corresponds to the standard inner product in $\mathbb{R}^{n^{2}}$. In particular, the Hermitian density operator $\rho = \sum_i \boldsymbol{\rho}_i F_i $,  is now represented by a real, $n^2$-dimensional vector $\boldsymbol{\rho}$.\footnote{This representation is closely related to the familiar ``coherence vector'' $\boldsymbol{v}$, wherein one represents $\rho$ as $\rho = \frac{1}{n}(\1+\sum_{j=1}^{n^2-1}\boldsymbol{v}_j F_j)$, where $F_{n^2}:=\1$ and the other orthonormal basis elements $\{F_j\}_{j=1}^{n^2-1}$ are in addition traceless; see, e.g., Ref.~\cite{byrd:062322}.} 
Since the Liouvillian is Hermitian preserving, i.e., $\left[\mathcal{L}\left(X\right)\right]^{\dagger}=\mathcal{L}\left(X^{\dagger}\right) \forall X$,
after coordinatization the operator $\mathcal{L}$ can be seen as
an operator $\mathbb{R}^{n^{2}}\mapsto\mathbb{R}^{n^{2}}$.
Indeed, denoting the corresponding matrix by $\boldsymbol{\mathcal{L}}$ in the chosen basis, i.e., $\boldsymbol{\mathcal{L}}_{ij}=\Tr\left(F_{i}\mathcal{L}\left(F_{j}\right)\right)$, one has

\begin{align}
\boldsymbol{\mathcal{L}}_{ij}^*&= \Tr\left[ (F_{i}\mathcal{L}\left(F_{j}\right) )^\dagger \right] = 
\Tr\left(\left[\mathcal{L}\left(F_{j}\right)\right]^{\dagger}F_{i}^{\dagger}\right)\notag \\
&=\Tr\left(\mathcal{L}\left(F_{j}\right)F_{i}\right)=\boldsymbol{\mathcal{L}}_{ij}\ .
\end{align}
i.e., the matrix $\boldsymbol{\mathcal{L}}$ is real and defines 
an operator $\mathbb{R}^{n^{2}}\mapsto\mathbb{R}^{n^{2}}$. Accordingly,  Eq.~\eqref{eq:general_ME} is transformed into a real-valued equation:
\beq
\dot{\boldsymbol{\rho}}=\boldsymbol{\mathcal{L}}\boldsymbol{\rho},\qquad\boldsymbol{\rho}(0)=\boldsymbol{\rho}_{0}\ .
\label{eq:rhobold}
\eeq
The cost~\eqref{nuovocosto} takes the form
\begin{equation}
J=\boldsymbol{C}^{T}\boldsymbol{\rho}(t_f)\ .
\label{coscos}
\end{equation}

We are now ready to state the PMP in the standard setting of real-valued functions, in the form needed for our purposes. Namely:

\begin{thm}
\label{adapt1} 
Assume $(\boldsymbol{\rho}^{*},s^{*})$ is an optimal pair for the problem defined by Eqs.~\eqref{eq:rhobold} and~\eqref{coscos} for a fixed final time $t_f$.\footnote{Recall that here the asterisk denotes optimality, not complex conjugation.} Then there exists a co-state vector $\boldsymbol{p}$ that satisfies\footnote{In the notation of Appendix~\ref{app:A} the co-state vector satisfies the adjoint equation~\eqref{3}, and in the present case $\frac{\partial f}{\partial x}=\boldsymbol{\mathcal{L}}$.}
\begin{equation}
\dot{\boldsymbol{p}}^{T}=-\boldsymbol{p}^{T}\boldsymbol{\mathcal{L}},\quad\boldsymbol{p}(t_f)=-\boldsymbol{C}\ .
\label{eq:co_general}
\end{equation}
Furthermore, define the PMP control Hamiltonian function 
\begin{equation}
\mathbb{H}\left(\boldsymbol{p},\boldsymbol{\rho},s\right)
=\boldsymbol{p}^{T}\mathcal{L}\boldsymbol{\rho}\ .
\label{HamilP'}
\end{equation}
We then have the \emph{maximum principle}: 
\begin{equation}
\mathbb{H}\left(\boldsymbol{p}(t),\boldsymbol{\rho}^*(t),s^*(t)\right)=\max_{v\in[0,1]}\mathbb{H}\left(\boldsymbol{p}(t),\boldsymbol{\rho}^*(t),v\right)\ ,
\label{MaxCon'}
\end{equation}
 and there exists a real non-negative constant $\lambda$ such that
\begin{equation}
\mathbb{H}\left(\boldsymbol{p}(t),\boldsymbol{\rho}^*(t),s^*(t)\right)=\lambda\ .
\label{forlanda'}
\end{equation}
\end{thm}

Theorem \ref{adapt0} is readily obtained applying coordinatization in reverse. 
To see that Eq.~\eqref{eq:co_general} corresponds to Eq.~\eqref{LiouvK1}, let us write Eq.~\eqref{eq:co_general} explicitly as $\dot{\boldsymbol{p}_{i}}=-\sum_{j}\boldsymbol{\mathcal{L}}_{ji}\boldsymbol{p}_{j}$.
Next, since $\boldsymbol{\mathcal{L}}$ is real we have $\boldsymbol{\mathcal{L}}_{ji}=\langle F_{j},\mathcal{L}F_{i}\rangle=\langle F_{j},\mathcal{L}F_{i}\rangle^*=\langle\mathcal{L}F_{i},F_{j}\rangle=\langle F_{i},\mathcal{L}^{\dag}F_{j}\rangle$.
Thus
\bes
\begin{align}
\dot{p} & =\sum_{i}\dot{\boldsymbol{p}_{i}}F_{i}=-\sum_{ij}\boldsymbol{\mathcal{L}}_{ji}\boldsymbol{p}_{j}F_{i} \\
&=-\sum_{ij}\langle F_{i},\mathcal{L}^{\dag}F_{j}\rangle\boldsymbol{p}_{j}F_{i}\\
 & =-\sum_{i}\langle F_{i},\mathcal{L}^{\dag}p\rangle F_{i}=-\mathcal{L}^{\dag}p\ .
\end{align}
\ees
Finally, the correspondence between Eqs.~\eqref{HamilP'} and~\eqref{HamilP} is a direct consequence of Eq.~\eqref{eq:<A,B>}: $\langle p,\mathcal{L}\rho\rangle = \boldsymbol{p}^{T}\mathcal{L}\boldsymbol{\rho}$.

\section{Proof of various formulas}
\label{app:calcs}

\subsection{Proof of Eq.~\eqref{eq:KXanti}}
\label{app:KXanti}

The proof is, for for arbitrary operators $A,B,X$:
\bes
\begin{align}
\<A,\mathcal{K}_{X}^{\dag}(B)\> &= \<\mathcal{K}_{X}(A),B\> \\
&= \<-i[X,A],B\> = i\Tr([X,A]^\dag B) \\
&= i \Tr(A^\dag X^\dag B-A^\dag B X^\dag ) \\
& = i\Tr(A^\dag [X^\dag,B]) = i\<A,[X^\dag,B]\> \\
& = \<A,-\mathcal{K}_{X^\dag}(B)\>\ ,
\end{align}
\ees
where we used $\<X,Y\> = \Tr[X^\dag Y]$.

\subsection{Proof of Eq.~\eqref{eq:Ddag-Red}}
\label{app:Ddag-Red}

Let $\mathcal{D}\rho := [W\rho,V^\dag]+[W,\rho V^\dag]$. Then:
\bes
\begin{align}
\<X,\mathcal{D}\rho\> &= \Tr\big(X^\dag [W\r,V^\dag_{\a}]+X^\dag [V,\r W^\dag]\big) \nonumber \\
&=\Tr\Big[ \big ( V^\dag X^\dag W - X^\dag V^\dag W \notag \\
&\qquad + W^\dag X^\dag V - W^\dag V X^\dag \big)\rho \Big] \\
&=\Tr\Big[ \big ( W^\dagger X V   - W^\dagger V X \notag \\
&\qquad +  V^\dag X W- X V^\dag W \big)^\dagger\rho \Big] \\
 &= \Tr\Big[\big( W^\dagger [X,V] + [V^\dag,X]W \big)^\dag\rho \Big] \\
& = \<\mathcal{D}^\dag X,\rho\>\ ,
\end{align}
\ees
which yields Eq.~\eqref{eq:Ddag-Red} when we replace $W$ by $W_{\a\b}$, $V$ by $S_\a = S_\a^\dag$, and sum over $\a,\b$.

\subsection{Proof of $v_{3}(t)\geq 0$ for sufficiently small $t$}
\label{app:v3>=0}

Let us compute the Dyson series solution of Eq.~\eqref{eq:Bloch-eq} to second order: 
\begin{align}
\boldsymbol{v}(t) &= \Big[\1 +\int_0^{t}dt_1\  \boldsymbol{\mathcal{M}}(s(t_1)) + \\
 &\int_0^{t}dt_1 \int_0^{t_1}dt_2\  \boldsymbol{\mathcal{M}}(s(t_1))\boldsymbol{\mathcal{M}}(s(t_2))+O(t^3)\big]\boldsymbol{v}_{0}\ . \notag
\end{align}
Using Eq.~\eqref{eq:M_Bloch}, for the initial condition $\boldsymbol{v}_{0}=(1,0,0)^{T}$ there is no contribution from the first (and in fact also the third) order, while the second order contributes via $\left[\boldsymbol{\mathcal{M}}(s(t_1))\boldsymbol{\mathcal{M}}(s(t_2))\right]_{31} = s(t_1)(1-s(t_2)) \geq 0$. Hence, using $s(t)\in[0,1]$, $v_3(t) = \int_0^{t}dt_1 \int_0^{t_1}dt_2\ s(t_1)(1-s(t_2)) + O(t^4) \geq 0$ for sufficiently small $t$. In contrast, for the initial condition $\boldsymbol{v}_{0}=(-1,0,0)^{T}$ we have $v_3(t) \leq 0$ by the same argument.

\begin{figure*}
\subfigure[\ ]{\includegraphics[width=7cm]{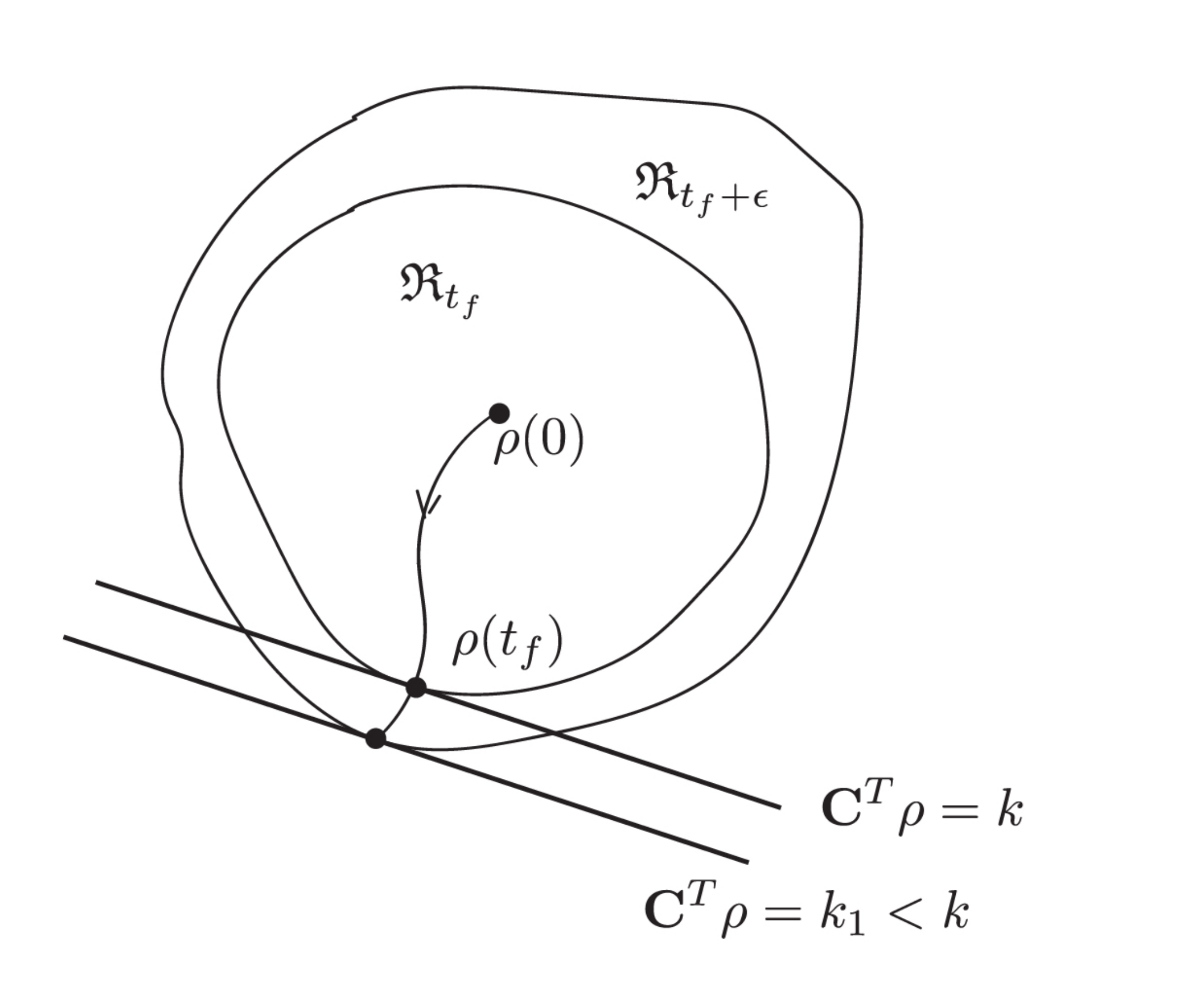}\label{FiguAD1}}
\subfigure[\ ]{\includegraphics[width=7cm]{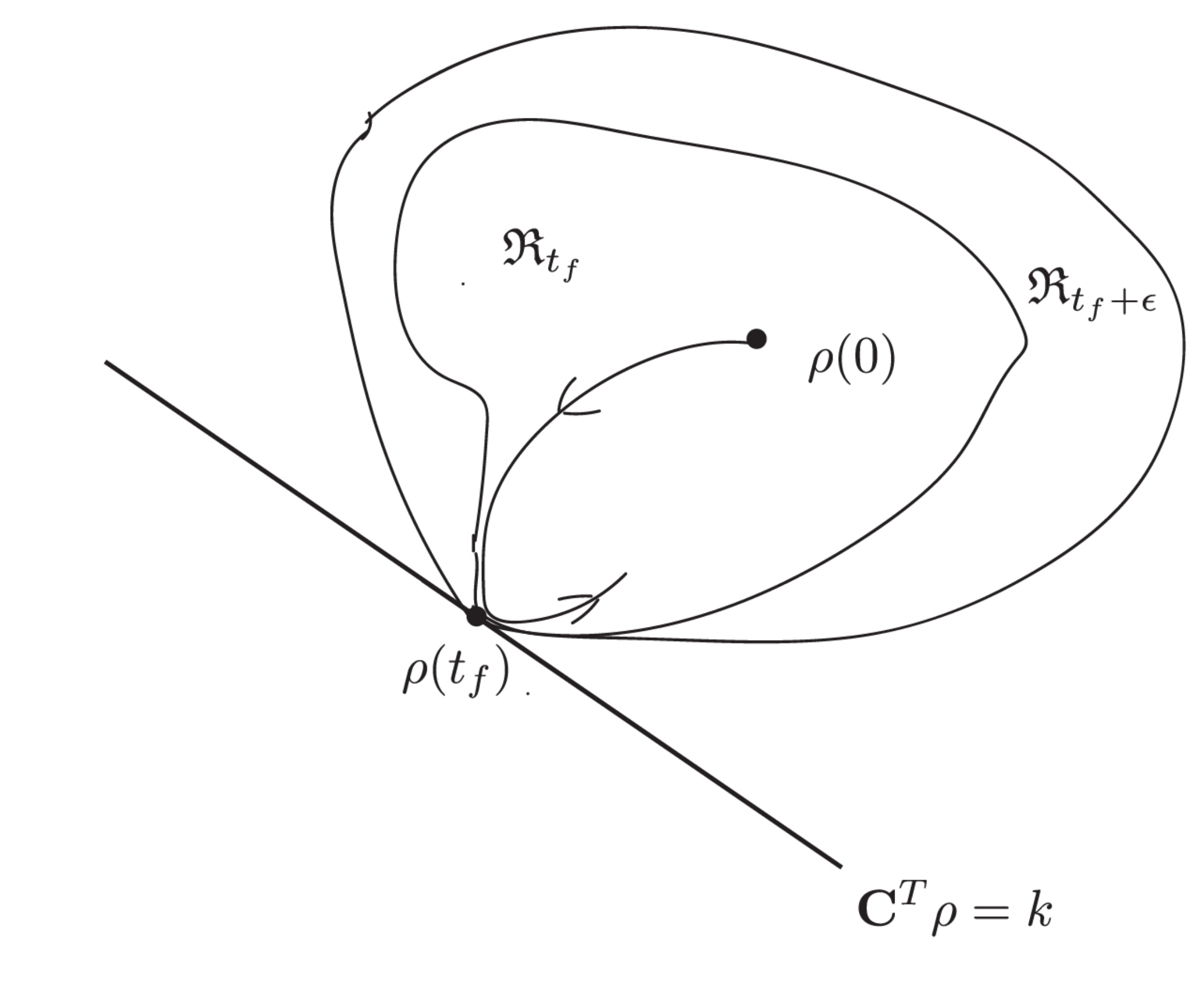}\label{FiguAD2}}
\caption{(a) Behavior of the optimal cost in the regular case where the reachable set increases with the final time at the optimal point $\boldsymbol{\rho}(t_f)$. In this case we expect $\lambda >0$.  (b) Behavior of the optimal cost in the  case where the reachable set increases with the final time but not at the optimal final  point $\boldsymbol{\rho}(t_f)$. In this case we expect $\lambda=0$.}
\end{figure*}

\section{Optimal control and the geometry of the reachable set}
\label{app:geomfig}

The cost considered in this work is linear in the state $\rho$, which after coordinatization we have identified with a point $\boldsymbol{\rho}$  in $\mathbb{R}^{n^2}$, that is, $J:={\bf C}^T \boldsymbol{\rho}$. In $\mathbb{R}^{n^2}$,  we also consider the reachable sets $\mathfrak{R}_t$ at various times $t$. It is of interest to consider the level lines (hyperplanes) ${\bf C}^T \boldsymbol{\rho}=k$ for various $k$'s. If ${ k}$ is the minimum cost at the final time $t_f$,  the level line  ${\bf C}^T \boldsymbol{\rho}=k$ intersects the boundary of the reachable set $\mathfrak{R}_{t_f}$ at the  point $\boldsymbol{\rho}(t_f)$. In our case, the reachable sets $\mathfrak{R}_t$ are always nondecreasing with $t$. Figure~\ref{FiguAD1} describes the regular situation of an active time constraint ($\lambda > 0$ in the main text). The intersection occurs at a point where the reachable set is increasing with time. 
Therefore an increase (decrease) of the final time $t_f$ results in a decrease (increase) of the optimal cost. However, in principle, a different situation  may  occur which is described in Fig.~\ref{FiguAD2}. 
In this case,  the reachable set $\mathfrak{R}_t$ increases with $t$ but not at  the point where the optimum occurs. In this case the final time constraint is not active ($\lambda=0$). 
Notice that by continuity of the reachable set $\mathfrak{R}_t$ with $t$, the point $\boldsymbol{\rho}(t_f)$ where the optimum was achieved with the final time $t_f$ will give the optimal for $t_f+\epsilon$ for sufficiently small $\epsilon >0$. 
The corresponding control will be a zero control which keeps the state at the initial value for time $\epsilon$ followed by the same control applied to reach $\boldsymbol{\rho}(t_f)$. 

Since we have not claimed uniqueness of the optimal control, the two situations may occur simultaneously for two different optimal trajectories. For one of them the final time constraint is active, while for the other one it is not. Given the additional structure of our problem, it should be possible to say more about the geometry of the reachable sets for the systems of interest here beside what is known from, for instance, Ref.~\cite{ayala_about_2017}. However, this is beyond the scope of this work.  

\section{Proof of Theorem~\ref{Active}}
\label{sec:proof-thm3}

\begin{proof}
Let us denote by $\tilde J(t_f)$ the cost function obtained with control $\tilde s:=\tilde s(t)$. By definition, the optimal control satisfies $J_{\min}(t_f) \leq \tilde J(t_f)$. Hence, to prove $J_{\min}(t_f) < J_{\min}(0)=\<C, \rho_0\>$, it suffices to show that there exists a control $\tilde s$ for Eq.~\eqref{eq:general_ME} such that the corresponding cost $\tilde J(t_f)$ satisfies $\tilde J(t_f) < J_{\min}(0)$. This is the proof strategy we employ here. 

Specifically, we consider a bang-bang control schedule $\tilde{s}$ in the interval
$[0,2t]$ and denote the corresponding cost starting from $\rho_{0}$ by $\tilde{J}(2t)$. We will
show that $\tilde{J}(2t)<\tilde{J}(0)$ for sufficiently small $t$, which gives
\begin{equation}
J_{\min}(2t)\leq\tilde{J}(2t)<\tilde{J}(0)=J_{\min}(0)=\Tr(C\rho_{0})\ ,
\label{Opla}
\end{equation}
and this proves the theorem with $t_f=2t$.

The class of controls we consider is $\tilde{s} \equiv1$ [corresponding
to $\mathcal{K}_{C}$ in Eq.~\eqref{eq:vonNeum}] for an interval of length $t$, followed by
$\tilde{s}\equiv0$ [corresponding to $\mathcal{K}_{B}$ in Eq.~\eqref{eq:vonNeum}]
for a second interval of length $t$. This gives for ${J}(t)$ [Eq.~\eqref{nuovocosto}]:
\begin{equation}
\tilde{J}(2t)=\Tr\left(Ce^{-iBt}e^{-iCt}\rho_{0}e^{iCt}e^{iBt}\right)\ .
\label{Jhat}
\end{equation}
We work in a basis where $B$ is diagonal with eigenvalues in decreasing
order: $B=\texttt{diag}(\lambda_{n},\lambda_{n-1},\dots ,\lambda_{1})$
and $\lambda_{1}<\lambda_{j}$, for each $j=2,3,\dots ,n$ (nondegeneracy).
In this basis $\rho_{0}=\texttt{diag}(0,\dots ,0,1)$.

In Eq.~\eqref{Jhat}, the Baker-Campbell-Hausdorff (BCH) formula yields: 
\begin{align}
&\tilde{J}(2t)=\Tr\Big[e^{iBt}Ce^{-iBt} \times\notag\\
&\qquad \Big(\rho_{0}-i[C,\rho_{0}]t-[C,[C,\rho_{0}]]\frac{t^{2}}{2}+O(t^{3})\Big)\Big] 
\label{Jhatplus}
\end{align}
Applying the BCH formula again, this time to $e^{iBt}Ce^{-iBt}$,
we obtain:
\begin{multline}
\tilde{J}(2t)=\Tr\Bigg[\left(C+i[B,C]t-[B,[B,C]]\frac{t^{2}}{2}\right)\times  \\
\left(\rho_{0}-i[C,\rho_{0}]t-[C,[C,\rho_{0}]]\frac{t^{2}}{2}\right)\Bigg]+O(t^{3}).
\end{multline}
Expanding, we obtain: 
\begin{align}
\tilde{J}(2t) & =\Tr(C\rho_{0})+i\Tr([B,C]\rho_{0})t-\Tr([B,[B,C]]\rho_{0})\frac{t^{2}}{2}\notag\\
 & \quad -i\Tr(C[C,\rho_{0}])t+\Tr\left([B,C][C,\rho_{0}]\right)t^{2}\notag\\
 & \quad -\Tr\left(C[C,[C,\rho_{0}]]\right)\frac{t^{2}}{2}+O(t^{3}).
\end{align}
Several of the terms in the above equation vanish. In particular,
\beq
\Tr\left([B,C]\rho_{0}\right)=\Tr\left([\rho_{0},B]C\right)=0\ ,
\eeq
since $B$ and $\rho_{0}$ commute. $\Tr([B,[B,C]]\rho_{0})=0$ for
the same reason. $-i\Tr(C[C,\rho_{0}])=-i\Tr(\rho_{0}[C,C])=0$, and
$\Tr\left(C[C,[C,\rho_{0}]]\right)=\Tr\left([C,\rho_{0}][C,C]\right)=0$.
Therefore we have 
\begin{equation}
\tilde{J}(2t)-\Tr(C\rho_{0})=\Tr\left([B,C][C,\rho_{0}]\right)t^{2}+O(t^{3})\ .
\label{poi}
\end{equation}
Write 
\begin{equation}
B=\begin{pmatrix}\Lambda & 0\\
0 & \lambda_{1}
\end{pmatrix},\qquad C=\begin{pmatrix}C_{1} & a\\
a^{\dagger} & c
\end{pmatrix}\ ,
\label{BandC}
\end{equation}
with $\Lambda=\texttt{diag}(\lambda_{n},\dots ,\lambda_{2})$, $C_{1}$
an $(n-1)\times(n-1)$ Hermitian matrix, $c$ a real number and $a$
an $(n-1)$-th dimensional complex vector. With these notations, we
have 
\begin{align}
[B,C] & =\begin{pmatrix}[\Lambda,C_{1}] & (\Lambda-\lambda_{1}\1)a\\
a^{\dagger}(\lambda_{1}\1-\Lambda) & 0
\end{pmatrix}\\{}
[C,\rho_{0}] & =\begin{pmatrix}0 & a\\
-a^{\dagger} & 0
\end{pmatrix}.
\end{align}
From this we obtain: 
\beq
\Tr\left([B,C][C,\rho_{0}]\right)=2a^{\dagger}(\lambda_{1}\1-\Lambda)a\ ,
\eeq
so that we have, from Eq.~\eqref{poi}:
\beq
\tilde{J}(2t)-\Tr(C\rho_{0})=-t^{2}\left(\sum_{j=2}^{n}(\lambda_{j}-\lambda_{1})|a_{j}|^{2}\right)+O(t^{3})\ ,
\eeq
where $a_{j}$ are the components of $a$. Since $\lambda_{j}>\lambda_{1}$
for each $j$, we have for sufficiently small $t$:
\beq
\tilde{J}(2t)-\Tr(C\rho_{0})=\tilde{J}(t)-J_{\min}(0)<0\ ,
\eeq
as required. We have assumed here that at least
one of the components of $a$ is nonzero. If that were
not the case then $[\rho_{0},C]=0$ and $\rho_0$ would be fixed not just under $B$ but also under $C$. There would then be no dynamics, which is a case that is naturally excluded.
\end{proof}

\section{Singular arcs}
\label{app:SA}

Along singular arcs we have $x_{C}\equiv x_{B}$, i.e.,
\begin{equation}
\langle B,S\rangle\equiv\langle C,S\rangle\ .
\label{condizione1}
\end{equation}
Differentiating Eq.~\eqref{condizione1}, using Eq.~\eqref{eq:S_ODE}
we find $\langle B,\mathcal{K}_{H}S\rangle=\langle C,\mathcal{K}_{H}S\rangle$.
Using $H = s B + (1-s) C$ and the antihermiticity of $\mathcal{K}_{H}$ we thus obtain:
\bes
\label{condizione2}
\begin{align}
 \<\mc{K}_H^\dag(B),S\> &= \<\mc{K}_H^\dag(C),S\> \Longrightarrow \\
 \<(1-s)[C,B]^\dag,S\> &= \<s[B,C]^\dag,S\> \Longrightarrow \\ 
 -(1-s)\<[C,B],S\> &= s\<[C,B],S\> \  \Longrightarrow\ \\
 \<[C,B],S\> &= 0 \ .
\end{align}
\ees

Analogously, differentiating Eq.~\eqref{condizione2} and using the antihermiticity of $\mathcal{K}_{H}$ 
again, setting $D:=[C,B]$, we have:
\bes
\label{condizione3}
\begin{align}
0 &= \<D,\dot{S}\> = \< D,\mc{K}_H(S)\> = \<\mc{K}_H^\dag(D),S\> \\
& = -(1-s) \< [C,D]^\dag,S\> - s \<[B,D]^\dag,S\> \\
& = (1-s) \< [C,D],S\> + s \<[B,D],S\>\ .
\end{align}
\ees
Conditions~\eqref{condizione1}-\eqref{condizione3} have to hold along a singular arc. In an algorithm to calculate the \emph{dynamical
Lie algebra} for the controllability of Eq.~\eqref{eq:vonNeum}
\cite{dalessandro_introduction_2007}, the matrices $B$ and $C$
are the matrices of ``depth'' zero in the calculation via iterated Lie
brackets. The matrix $[C,B]$ is of depth one and the matrices $[[C,B],C]$ and
$[[C,B],B]$ are of depth two. Now, one can have either (i) $\langle [C,D], S \rangle \neq \langle [B,D], S \rangle $ or (ii) $\langle [C,D], S \rangle = \langle [B,D], S \rangle$. In case (i) it follows from Eq.~\eqref{condizione3} that
\beq 
s=\frac{\langle [C,D], S \rangle}{\langle [C,D],S \rangle -\langle [B,D], S \rangle} \label{eq:s_continuous}
\eeq 
(compare with Ref.~\cite[Eq.~(12)]{brady_optimal_2021}). This shows the continuity of $s$ in the corresponding open set(s). 
Case (ii) implies $\langle [C,D],S \rangle \equiv \langle [B,D],S \rangle \equiv 0$ (in some closed set). One can further differentiate one of these equations and obtain an analog of  Eq.~\eqref{condizione3} at a higher order, at which point similar reasoning can be applied.  In principle $s$ can be defined in different intervals by equations such as Eq.~\eqref{eq:s_continuous} or its higher order generalizations. In each interval $s$ is continuous because of the continuity of $S$. We leave a more general proof of continuity of $s$ on the \emph{entire} singular arc as an open problem.
In any case, equations~\eqref{condizione2}-\eqref{condizione3} provide
information on the dynamics along singular arc intervals. They are used in the example discussed in Appendix~\ref{low}.

\section{Optimal control protocol for the spin-$\frac{1}{2}$ model}
\label{low}

Here we analyze in detail the optimal control problem for the spin-$\frac{1}{2}$
model treated in Sec.~\ref{esempio}. Our goal is to give a simple but
explicit example to show how the results
developed in this paper can be used to find the optimal control. We
shall use the same notation as in the example of Sec.~\ref{esempio}.
To avoid the situation of an inactive terminal time constraint described
in the example, we assume that the initial state is the ground state $x_{0}=(-1,0,0)^{T}$,
so that we can apply Theorem~\ref{Active}.

\subsection{The global minimum is found using two non-singular arcs in time $t_f \geq \pi$}

Recall that $C=\sigma^{z}/2$ and $B=\sigma^{x}/2$.
The global minimum of the cost~\eqref{nuovocosto} is $J_{\min}=\Tr[C\rho_f] = -1/2$, achieved when $\rho_f = (\1-\sigma^{z})/2$, the ground state of $C=\sigma^z/2$. Given that our initial condition is $\rho_{0}=(\1-\sigma^{x})/2$ (the ground state of $B$,
corresponding to $\boldsymbol{v}=(-1,0,0)^{T}$), we can trivially reach $\rho_f=(\1-\sigma^{z})/2$ by applying two consecutive bangs (i.e., unitary single-qubit gates): first $e^{-i (\pi/2) C}$ (rotation to $(\1-\sigma^{y})/2$) with $s\equiv0$, then $e^{-i (\pi/2) B}$ (rotation to $(\1-\sigma^{z})/2$) with $s\equiv1$. Each bang lasts for a time $\frac{\pi}{2}$,
therefore the total bang-bang sequence last for a total time of $\pi$. \emph{This sequence presents no singular arcs}. For any $t_f>\pi$ the constraint on $t_f$ becomes inactive, i.e., increasing $t_f$ cannot further lower the value of  $J_{\min}$. Since we assume that the global ground state is not
reached (recall the discussion in Sec.~\ref{sec:active-constraint}), henceforth we assume that $t_f<\pi$. In principle, this setting could still allow for the appearance of singular arcs. However, we shall show that this is not the case.


\subsection{Conditions on the singular arcs for the spin-$1/2$ model}

Let us derive the conditions on the singular arcs in the present problem, which are a special case of the computations carried out in Appendix~\ref{app:SA}. Using Eq.~\eqref{commurel}, we obtain $\left[C,B\right]=i\sigma^{y}/2$,
$\left[[C,B],C\right]=-\sigma^{x}/2$, $\left[[C,B],B\right]=\sigma^{z}/2$.
Using these, Eqs.~\eqref{eq:x_C-b}, \eqref{eq:x_B}, and~\eqref{condizione1} for the switching operator
$S$ become: 
\begin{equation}
2x_{C}=\Tr\left(S\sigma^{z}\right)\equiv\Tr\left(S\sigma^{x}\right)=2x_{B}\ .
\label{newcondizione1}
\end{equation}
Condition~\eqref{condizione2} becomes: 
\begin{equation}
\Tr\left(S\sigma^{y}\right)\equiv0\ ,
\label{condizione5}
\end{equation}
and condition~\eqref{condizione3} becomes $(1-s)\Tr(\sigma^{x}S)-s\Tr(\sigma^{z}S)\equiv0$,
which using Eq.~\eqref{newcondizione1} gives: 
\begin{equation}
\left(1-2s\right)\Tr\left(\sigma^{x}S\right)\equiv 0\ .
\label{condizione4}
\end{equation}
Thus, either $s=1/2$ or $\Tr(\sigma^{x}S)\equiv0$. Let assume the latter. From Eqs.~\eqref{newcondizione1}
and~\eqref{condizione5} we obtain $\Tr(\sigma^{z}S)=\Tr(\sigma^{y}S)\equiv0$.
Since we can expand $S=\frac{1}{2}\sum_{i=1}^3 \Tr(S\sigma_i)\sigma_i$ ($S$ is traceless since it is defined as a commutator), this would then imply that $S\equiv0$ on a singular interval. However, since $S$ satisfies
the linear equation~\eqref{eq:S_ODE}, this would imply $S\equiv0$ on
the whole interval $[0,t_f]$ and, in particular, $[p_0,\rho_{0}]=[p_f,\rho_f]=-[C,\rho_f]=0$.
This would imply that $\rho_f$ is a linear combination of eigenprojectors
of $C$, but since $\rho_f$ is a pure state it must in fact be equal to a single eigenprojector. Moreover, this must be the ground state of $C$ since $J(t)$ is minimized at $t=t_f$. But, since $J(t_f)\le J(0)$, the only possibility is that $J=\Tr(C\rho)$ reaches its
global minimum at $t_f$, which contradicts our assumption that $t_f$ is smaller than a value that would allow the global minimum to be reached. Hence we conclude that $\Tr(\sigma^{x}S)\ne 0$ in Eq.~\eqref{condizione4}, \emph{which yields $s\equiv1/2$ on the singular arcs}. However, we shall see in Proposition~\ref{exsing} that singular arcs are in fact not possible in this case.

\subsection{Candidate optimal controls with a singular arc}
\label{candid}

Using conditions~\eqref{newcondizione1} and~\eqref{condizione5} along
with Eq.~\eqref{eq:new_controlH} equated to $\lambda$, we have that
in the time interval of a singular arc 
\begin{equation}
S\equiv\lambda\left(\sigma^{x}+\sigma^{z}\right)\ .
\label{valuesa}
\end{equation}
$\lambda=0$ is impossible because according to the
argument at the end of the previous subsection $S\equiv0$
is to be excluded. Since $\lambda\ne 0$, we can apply all the conclusions
of Theorem~\ref{Bradygen} and affirm that the optimal control starts
with an $s\equiv0$ bang arc and ends with an $s\equiv1$ bang arc.
Therefore, preceding or following a singular arc we must have $s\equiv0$
or $s\equiv1$, respectively. Let us show that after a singular arc we cannot go
to a switching point, i.e., where Eq.~\eqref{newcondizione1}
holds [$(\lambda,\lambda)$ in Fig.~\ref{switchfig}]. (Analogously, changing the sign of time, we
can show that a singular arc cannot be preceded by a switching point.)
Assume that after the singular arc we have $s\equiv0$. The switching
operator $S=S(t)$, with $t=0$ at the end of this singular arc, is then the solution of Eq.~\eqref{eq:S_ODE} with the initial condition~\eqref{valuesa}, i.e., $S(t)=\lambda e^{-it\sigma^{z}/2}\sigma^{x}e^{it\sigma^{z}/2}+\lambda\sigma^{z}$. The minimum time needed for it to return to a switching point is $t=2\pi$. This contradicts the fact that $t_f<\pi$
and therefore is impossible. Similar reasoning shows that we
cannot go back to a switching point with $s\equiv1$. Therefore, we
have learned the following fact about optimal control in the single-qubit
case:

\begin{prop}
The optimal control has at most one singular arc and, if it does,
the optimal control is the sequence $s\equiv0$, $s\equiv\frac{1}{2}$,
$s\equiv1$. 
\end{prop}

Consider now the initial switching operator $S_{0}$, which together
with the differential equation Eq.~\eqref{eq:S_ODE} determines the control
sequence. Since $S$ is traceless, we can write $S_{0}=\boldsymbol{r}_{0}\cdot\boldsymbol{\sigma}$,
but using Eq.~\eqref{sw} and $\rho_0=\frac{1}{2}({\bf \1}-\sigma^{x})$ we see by expanding $p$ in the Pauli matrix basis that $S_0$ cannot contain $\sigma^x$, i.e., 
we find that $S_{0}$ has the form 
\begin{equation}
S_{0}=r_{0y}\sigma^{y}+r_{0z}\sigma^{z}\ .
\label{S0}
\end{equation}
In the first interval $s\equiv0$ and therefore,
from Eq.~\eqref{eq:S_ODE}:
\bes
 \label{firstint}
\begin{align}
S(t) & =r_{0z}\sigma^{z}+r_{0y}e^{-it\sigma^{z}/2}\sigma^{y}e^{it\sigma^{z}/2} \\
 & =r_{0z}\sigma^{z}+r_{0y}\cos(t)\sigma^{y}-r_{0y}\sin(t)\sigma^{x}\ .
 \label{firstint-b}
 \end{align}
\ees

If there is a singular arc and therefore $S$ takes the form~\eqref{valuesa}, then
we must have $t=\frac{\pi}{2}$ and $r_{0y}=-r_{0z}$ or $t=\frac{3\pi}{2}$
and $r_{0y}=r_{0z}$. The second case is to be excluded since $t_f<\pi$.
After the singular arc we would have $s\equiv1$, which, using Eq.~\eqref{eq:S_ODE} again would give: 
\bes
\begin{align}
S(t)&=r_{0z}\sigma^{x}+r_{0z}e^{-it\sigma^{x}/2}\sigma^{z}e^{it\sigma^{x}/2}\\
&=r_{0z}\sigma^{x}+r_{0z}(\cos(t)\sigma^{z}-\sin(t)\sigma^{y})\ .
\label{eq:F7b}
\end{align}
\ees
Since we have to reach the point $(x_C,x_B)=(0,\lambda)$ in Fig.~\ref{switchfig}, we must have $\cos\left(t\right)=0$, i.e., \textbf{$t=\frac{\pi}{2}$
}or $t=\frac{3\pi}{2}$, which has to be added to the time used before
the last interval. Therefore the total time is greater than or equal
to $\pi$, which we have excluded.

In conclusion we have:
\begin{prop}
\label{exsing} No singular arc exists in the optimal control for
the spin-$1/2$ example with $t_f<\pi$. 
\end{prop}

\subsection{Candidate optimal controls without singular arcs}

Now we consider the optimal control candidates knowing that they must be free of singular arcs, i.e., they can consist only of bangs. Since $S(t)$ is traceless we again use the 
parametrization $S(t)=\boldsymbol{r}(t)\cdot\boldsymbol{\sigma}$.
We already know [Eq.~\eqref{S0}] that:
\begin{equation}
\boldsymbol{r}(0)=\boldsymbol{r}_{0}=(0,r_{0y},r_{0z})^{T}\ .
\label{s0vect}
\end{equation}
We know from Eq.~\eqref{eq:S_ODE} that the vector $\boldsymbol{r}$ evolves according to Eqs.~\eqref{eq:Bloch-eq}-\eqref{eq:M_Bloch}. More explicitly, let $\boldsymbol{X}:=\boldsymbol{\mathcal{K}}_{B}$ and $\boldsymbol{Z}:=\boldsymbol{\mathcal{K}}_{C}$.
When $s\equiv0$, $H=C$ and $\boldsymbol{r}$ evolves according to
$e^{t\boldsymbol{Z}}$, and likewise when $s\equiv1$ it evolves according to $e^{t\boldsymbol{X}}$, where, using Eq.~\eqref{adcadb}:
\bes
\label{ext}
\begin{align}
\label{extX}
e^{t\boldsymbol{X}}&=\left(\begin{array}{ccc}
1 & 0 & 0\\
0 & \cos(t) & -\sin(t)\\
0 & \sin(t) & \cos(t)
\end{array}\right)\\
\label{extZ}
e^{t\boldsymbol{Z}}&=\left(\begin{array}{ccc}
\cos(t) & -\sin(t) & 0\\
\sin(t) & \cos(t) & 0\\
0 & 0 & 1
\end{array}\right)\ .
\end{align}
\ees
Since we have shown that $\lambda>0$, from Theorem~\ref{Bradygen},
the control law will start with an $s\equiv0$ bang arc and end with
an $s\equiv1$ bang arc. The control law is determined by a sequence
of intervals of lengths $\{\tau_{1},\tau_{2},\dots\}$ where for $k$ odd
(even) $\tau_{k}$ marks the switch from $s\equiv0$ to $s\equiv1$ ($s\equiv 1$ to $s\equiv 0$),
that is $\boldsymbol{Z}\rightarrow\boldsymbol{X}$ ($\boldsymbol{X}\rightarrow\boldsymbol{Z}$). That is, 
\beq
\label{eq:recur-even-odd}
\boldsymbol{r}(t_k) = 
\left\{
\begin{array}{c}
  e^{\tau_k\boldsymbol{Z}}\boldsymbol{r}(t_{k-1})\quad k\text{  odd}   \\
  e^{\tau_k\boldsymbol{X}}\boldsymbol{r}(t_{k-1})\quad k\text{  even}
\end{array}
\right.\ ,
\eeq
where $t_k = \sum_{i=1}^k\tau_i$ is the total time after $k$ intervals.
Note that in principle
there is no guarantee that such a switching sequence is finite, even if the
total control interval is finite; 
this is known in the control theory literature as the \emph{Fuller
phenomenon} (see, e.g., Ref.~\cite{borisov_fullers_2000}). We shall see
in Remark \ref{Values} that this does not happen in our case and
we have a \emph{finite} sequence of intervals of lengths $\{\tau_{1},\tau_{2},\dots ,\tau_{N}\}$
with $N$ even (according to Theorem~\ref{Bradygen}). Given our definitions,
$t_f=\sum_{i=1}^N \tau_i$, where $\tau_{1}$ and $\tau_{N}$ are the lengths
of the initial ($\boldsymbol{Z}$) and final ($\boldsymbol{X}$) arcs,
respectively.

\subsection{Characterization of the switching times}

Note that the vector $\boldsymbol{r}=(r_x,r_y,r_z)^T$ consists of the components of $S$ along the Pauli basis, and that $r_x=x_B$ and $r_z=x_C$ [Eq.~\eqref{newcondizione1}]. Recall also that, as argued in Sec.~\ref{sec:switch-eq} (see Fig.~\ref{switchfig}), $(x_{C},x_{B})=(\lambda,\lambda)$ at every switching point. Hence $r_x=r_z=\lambda$ at every switching point between nonsingular arcs in our discussion below.

The optimal candidate control law is characterized by a sequence of
intervals of length $\tau_{1}$, $\tau_{2}$, etc. Define the
sequence $\{\Delta_{k}\}$ recursively from the sequence $\{\tau_{k}\}$ via $\Delta_{0}=0$,
$\Delta_{k}=\tau_{k}-\Delta_{k-1}$, for $k=1,2,\dots$. Then:
\begin{lem}
\label{recur1} 
For $n=1,2,\dots $, except 
for the $n$ corresponding
to the last control interval 
\begin{equation}
\left(-1\right)^{n}\sin\left(\Delta_{n}\right)r_{0y}=r_{0z},
\label{recur2}
\end{equation}
[cf. Eq.~\eqref{S0}] and 
\begin{equation}
\boldsymbol{r}\left(t_{n}\right)=
\left(r_{0z},
\cos(\Delta_{n})r_{0y},
r_{0z}\right)^T
\ .
\label{recur3}
\end{equation}
 
\end{lem}

\begin{proof}
The proof is by induction on $n$. For $n=1$ we have $\Delta_{1}=t_{1}=\tau_1$. At the end of the first arc we must reach the point $(x_{C},x_{B})=(\lambda,\lambda)$, which means that $\boldsymbol{r}(\tau_{1})=(r_{0z},*,r_{0z})^{T}$. Thus, using Eq.~\eqref{eq:recur-even-odd}
and $e^{\tau_{1}\boldsymbol{Z}}$ in Eq.~\eqref{extZ}, we obtain Eqs.~\eqref{recur2}
and~\eqref{recur3}.

Now assume Eq.~\eqref{recur2} and Eq.~\eqref{recur3} hold for $n-1$.
If $n$ is even we have 
\begin{align}
\boldsymbol{r}\left(t_{n}\right) =e^{\tau_{n}\boldsymbol{X}}\boldsymbol{r}\left(t_{n-1}\right)\ .
\label{withX}
\end{align}
Next, use Eq.~\eqref{extX}, and impose $\boldsymbol{r}(t_{n})=(r_{0z},*,r_{0z})^{T}$, since $(x_{C},x_{B})=(\lambda,\lambda)$ at every switching point.
Equality of the $z$ component then gives $\sin(\tau_{n})\cos(\Delta_{n-1})r_{0y}+\cos(\tau_{n})r_{0z}=r_{0z}$,
and using Eq.~\eqref{recur2} with $n$ replaced by $n-1$, we obtain $r_{0z}=\sin(\tau_{n})\cos(\Delta_{n-1})r_{0y}-\cos(\tau_{n})\sin(\Delta_{n-1})r_{0y}=\sin(\tau_{n}-\Delta_{n-1})r_{0y}=\sin(\Delta_{n})r_{0y}$.
Calculating the $y$ component of $\boldsymbol{r}(t_n)$, we obtain $\cos(\tau_{n})\cos(\Delta_{n-1})r_{0y}-\sin(\tau_{n})r_{0z}=\cos(\tau_{n})\cos(\Delta_{n-1})r_{0y}+\sin(\tau_{n})\sin(\Delta_{n-1})r_{0y}=\cos(\Delta_{n})r_{0y}$,
using again the inductive assumption Eq.~\eqref{recur2}. A similar calculation
with $\boldsymbol{Z}$ replacing $\boldsymbol{X}$ in Eq.~\eqref{withX}
gives the result when $n$ is odd. 
\end{proof}

\subsection{Determination of the Optimal Control}

\label{final}

We now use the formulas in the above Lemma to determine the optimal
control. Define $\mu=\arcsin\left(\frac{r_{0z}}{r_{0y}}\right)$ and
notice that from Eq.~\eqref{recur2} for $n=1$ and $\Delta_{1}=\tau_{1}$
we have $\mu=\arcsin\left(-\sin\left(\tau_{1}\right)\right)$. Since $0<\tau_{1}<\pi$, we have $\mu=-\tau_1$ for $\tau_1\in (0,\pi/2]$ and $\mu=-\pi+\tau_1$ for $\tau_1\in [\pi/2,\pi)$, and in particular $\mu <0$.

Let us consider first the possibility that $0<\tau_{1}\leq\frac{\pi}{2}$.
We also have $\tau_{1}=\Delta_{1}=-\mu$. 
If there is more than one switch (i.e., $\tau_2>0$), then we 
can derive $\Delta_{2}$ from Eq.~\eqref{recur2}.
We have either $\Delta_{2}=\pi-\mu+2l\pi$ or $\Delta_{2}=\mu+2l\pi$
for integer $l$. Recalling that $\tau_{2}=\tau_{1}+\Delta_{2}$, in the
first case we have $\tau_{2}=\pi-2\mu+2l\pi$, and in the second case
$\tau_{2}=2l\pi$. The second case is not possible because $\tau_{2}$ must
be in $(0,\pi)$.
The first case is not possible either
because $l\geq0$ would contradict that the total time must be less
than $\pi$ while $l<0$ would give a negative or zero interval $\tau_{2}$. Therefore,
in the case $0<\tau_{1}\leq\frac{\pi}{2}$ there exists only one switch
and the control is simply the sequence of two bangs, one corresponding
to $s\equiv0$ followed by one corresponding to $s\equiv1$. Before
determining where the switch must occur, let us consider the case
$\frac{\pi}{2}<\tau_{1}<\pi$.

\begin{lem}
\label{nuovolemma} 
Assume that $\tau_1\in (\pi/2,\pi)$ and define $\mu=\arcsin\left(\frac{r_{0z}}{r_{0y}}\right)<0$.
Then $\Delta_{n}=\pi+\mu$ for $n$ odd and $\Delta_{n}=\mu$ for
$n$ even. 
\end{lem}

\begin{proof}
The claim follows by induction from Eq.~\eqref{recur2}. Applying
it for $n=1$, since $\Delta_{1}=\tau_{1}$, we have $\Delta_{1}=\pi+\mu$.
Now, assume that the claim is true for $n$ even. From Eq.~\eqref{recur2}
applied for $n$ even, we obtain $\Delta_{n}=\mu+2l\pi$ or $\Delta_{n}=\pi-\mu+2l\pi$,
for integer $l$. Since $\tau_{n}=\Delta_{n-1}+\Delta_{n}$, using the
inductive assumption, we obtain in the two cases, $\tau_{n}=\pi+2\mu+2l\pi$
and $\tau_{n}=2\pi+2l\pi$, respectively. The latter case is impossible
because it would mean that $\tau_{n}>\pi$.
The first case is only possible with $l=0$, because $l>0$ would imply
$\tau_{n}>2\pi$ while $l<0$ would give a negative time interval. Since we saw above that $\Delta_{n}=\mu+2l\pi$ for $n$ even, this gives
$\Delta_{n}=\mu$ for such $n$. 

Let us now prove that $\Delta_{n}=\pi+\mu$
for $n>1$ and odd. Again using Eq.~\eqref{recur2} we obtain
either $\Delta_{n}=-\mu+2l\pi$ or $\Delta_{n}=\pi+\mu+2l\pi$. The
first case is impossible because it would mean $\tau_{n}=\Delta_{n-1}+\Delta_{n}=2l\pi$
(using the inductive assumption). The second case would give $\tau_{n}=\Delta_{n-1}+\Delta_{n}=\pi+2\mu+2l\pi$
which is only possible for $l=0$. This gives $\Delta_{n}=\pi+\mu$.
\end{proof}

\begin{rem}
\label{Values} One of the consequences of the above lemma is that the switching sequence is finite, i.e.,  we do not have intervals between two switches which
become arbitrarily small and hence the Fuller phenomenon~\cite{borisov_fullers_2000}
is ruled out in our case.
In particular
if $0<\tau_{1}<\frac{\pi}{2}$ there is only one switch, as we have seen,
while if $\frac{\pi}{2}<\tau_{1}<\pi$ then we have multiple switches with $\tau_{k}=\Delta_{k}+\Delta_{k-1}=\pi+2\mu$, which is a constant independent of $k$.
\end{rem}

In order to learn more about the optimal control, and rule out the second case of $\frac{\pi}{2}<\tau_{1}<\pi$, we examine the final
arc, which is of the form $e^{\tau_{N}\boldsymbol{X}}$ ($s\equiv1$). Recall that with the final arc we have to reach the point $(x_C,x_B)=(0,\lambda)$, which imposes that
the final switching operator is of the form $\boldsymbol{r}(t_N)=(r_{0z},*,0)^{T}$ [see Eq.~\eqref{eq:F7b} and the discussion immediately below it]. 
Thus, using Eq.~\eqref{withX} with $n=N$ and Eq.~\eqref{extX}, 
we obtain $\sin(\tau_{N})\cos(\Delta_{N-1})r_{0y}+\cos(t_{N})r_{0z}=0$.
Using Eq.~\eqref{recur2} and $r_{0y}\ne 0$, we obtain $\sin(\tau_{N})\cos(\Delta_{N-1})-\sin(\Delta_{N-1})\cos(\tau_{N})=\sin(\Delta_{N})=0$,
where $\Delta_{N}=\tau_{N}-\Delta_{N-1}$. Therefore $\tau_{N}=\Delta_{N-1}+l\pi$
for $l$ integer. Now there are two cases: Multiple switches or only
one switch. In the case of multiple switches, we are in the situation
described in Lemma~\ref{nuovolemma}. We have $\Delta_{N-1}=\Delta_{1}=\pi+\mu$.
Therefore $\tau_{N}=\pi+\mu+l\pi$. The integer $l$ must be zero because
if it is positive we have $\tau_{N}>\pi$ and if it is negative, we have
a negative interval $\tau_{N}$. Therefore $\tau_{N}=\pi+\mu$. However $t_f\geq \tau_{1}+\tau_{N}=2\Delta_{1}=2(\pi+\mu)>\pi$
which is impossible. Therefore the situation $\frac{\pi}{2}<\tau_{1}<\pi$
cannot occur. The only possibility is the situation with $0<\tau_{1}<\frac{\pi}{2}$
with one switch only. In this case, as above we have $\tau_{2}=\Delta_{1}+l\pi=\tau_{1}+l\pi$
with $l=0$ since again $l<0$ will give a negative time interval and $l$ positive
will give total time greater than $\pi$. So $\tau_{2}=\tau_{1}$ and the
optimal control is the simplest one. This completes the proof of Proposition~\ref{conclusione}.

\section{Additional considerations regarding the shortening of the initial and final arcs}
\label{app:shortening}

In principle one can be more quantitative about the shortening of the arcs discussed in Sec.~\ref{shortening} by repeating the analysis
of Ref.~\cite{campos_venuti_recurrence_2015} for the quantity $x_{B}(\Delta t)$ [Eq.~\eqref{eq:switching_time}].
The latter work provided a detailed analysis of the
return probability $\mathcal{F}(t):=\left|\langle\psi|e^{-itH}|\psi\rangle\right|^{2}$, where $H$ is a time-independent Hamiltonian and $|\psi\>$ an initial state. One of the results of Ref.~\cite{campos_venuti_recurrence_2015}
was an explicit form for the average number of zeroes of the equation
$\mathcal{F}(t)=v$. Indeed, $\mathcal{F}$ is a particular case of
the left-hand side of Eq.~\eqref{eq:switching_time} and can be written
in that form with $B=S_{0}=\ketb{\psi}{\psi}$, $C=H$. Let
$N_{x_{B}}(v)$ be the average number of solutions of the equation
$x_{B}(t)=v$. The number of zeroes $N$ in a large interval of length $T$ turns
out to be proportional to $T$: $N_{x_{B}}(v)=TD_{x_{B}}(v)$ where
$D_{x_{B}}$ can be computed using the methods of Ref.~\cite{campos_venuti_recurrence_2015}.
Then $\Delta t\simeq T/N_{x_{B}}(\lambda)=1/D_{x_{B}}\left(\lambda\right)$. Using ~\cite[Eq.~(6)]{campos_venuti_recurrence_2015}, which applies for the special case mentioned above, we then obtain:
\beq
\Delta t(\lambda) \simeq \frac{\sqrt{\pi}}{2}\frac{1}{\Delta E}\sqrt{\frac{\<x_B\>}{\lambda}}e^{\lambda/\<x_B\>}\ ,
\label{eq:rec-time}
\eeq
where $\<x_B\>:=\lim_{T\to\infty}\frac{1}{T}\int_0^T x_B(t)dt$ and $\Delta E$ is the standard deviation of the energies $\{E_k\}$ with respect to the distribution $\{p_k = \left \lVert \Pi_k |\psi\>\right \rVert^4/\sum_n \left \lVert \Pi_n |\psi\>\right \rVert^4\}$ (recall that $C=\sum_{k}E_{k}\Pi_{k}$ is the spectral resolution of $C$). The derivation of Eq.~\eqref{eq:rec-time} is subtle and is carried out in Ref.~\cite{campos_venuti_recurrence_2015}. Since it only represents a special case in our context, we do not pursue it further here, and present Eq.~\eqref{eq:rec-time} mainly to stimulate the interest of the reader and establish a possible entry point towards a rigorous treatment of the shortening of the initial and final arcs.

As a final comment, note that when $C$ is a classical Ising Hamiltonian of the form $C=J \sum_{i,j=1}^{n}\sigma_{i}^{z}\sigma_{j}^{z}$, the energies are integer multiples of $J$ and highly degenerate (therefore the spectrum is commensurate: $\left\{ E_{k}\right\}  \subset \left\{ 0,\pm J,\pm2J,\ldots\right\}$).
This implies that the function $x_{B}(t)$ in Eq.~\eqref{eq:trigon_pol} is periodic with period not larger than $2\pi/J$ (as opposed to being
almost periodic) and this has implications for $D_{x_{B}}(\lambda)$, but the general considerations we have outlined still hold.



%

\end{document}